\documentclass[10pt, twocolumn]{IEEEtran}

\usepackage{psfrag}
\usepackage{amssymb}
\usepackage{amsmath}
\usepackage{pifont}
\usepackage{cite}
\usepackage{graphics}
\usepackage{graphicx}
\usepackage{epsfig}
\usepackage{subfigure}
\usepackage{tablefootnote}
\usepackage{url}
\usepackage{amscd}
\usepackage{enumerate}
\usepackage[colorlinks, citecolor=blue,linkcolor=blue]{hyperref}

\newtheorem{theorem}{Theorem}
\newtheorem{remark}{Remark}
\newtheorem{lemma}[theorem]{Lemma}

\newtheorem{proposition}[theorem]{Proposition}
\newtheorem{corollary}[theorem]{Corollary}

\begin{document}

\title{Recovery of Sparse Signals via Generalized Orthogonal Matching Pursuit: A New Analysis}
\date{}
\author{\IEEEauthorblockN{Jian Wang$^{1,2}$, Suhyuk Kwon$^3$, Ping Li$^{2}$ and Byonghyo Shim$^{3}$} \\
\IEEEauthorblockA{$^1$Department of Electrical \& Computer Engineering, Duke University  \\
$^2$Department of Statistics \& Biostatistics, Department of Computer Science, Rutgers University  \\
$^3$Department of Electrical \& Computer
Engineering, Seoul National University  \\
Email: jian.wang@duke.edu, pingli@stat.rutgers.edu, \{shkwon,bshim\}@snu.ac.kr}\\
 

} \maketitle

\maketitle

\begin{abstract}
As an extension of orthogonal matching pursuit (OMP) improving the
recovery performance of sparse signals, generalized OMP (gOMP) has
recently been studied in the literature.
In this paper, we present a new analysis of the gOMP algorithm using restricted isometry property (RIP). We show that if the measurement matrix $\mathbf{\Phi} \in \mathcal{R}^{m \times n}$
satisfies the RIP with
$$\delta_{\max \left\{9, S + 1 \right\}K} \leq \frac{1}{8},$$ then gOMP performs stable reconstruction of all $K$-sparse signals $\mathbf{x} \in \mathcal{R}^n$ from
the noisy measurements $\mathbf{y} = \mathbf{\Phi x} + \mathbf{v}$ within $\max \left\{K, \left\lfloor \frac{8K}{S} \right\rfloor \right\}$ iterations where $\mathbf{v}$ is the noise vector and $S$ is the number of indices chosen in each iteration of the gOMP algorithm. 
For Gaussian random measurements, our results indicate that
the number of required measurements is essentially $m = \mathcal{O}(K \log
\frac{n}{K})$, which is a significant improvement over the existing
result $m = \mathcal{O}(K^2 \log \frac{n}{K})$, especially for large $K$.
\end{abstract}


\IEEEpeerreviewmaketitle

\setcounter{page}{1}

\section{Introduction} \label{sec:int}
Orthogonal matching pursuit (OMP) is a greedy algorithm widely used
for the recovery of sparse signals~\cite{pati1993orthogonal,mallat1993matching,tropp2007signal,davenport2010analysis,cai2011orthogonal,wang2012Recovery,ding2013perturbation}.
The goal of OMP is to recover a $K$-sparse signal vector $\mathbf{x}
\in \mathcal{R}^n$ ($\|\mathbf{x}\|_0 \leq K$) from its linear measurements
\begin{equation}
  \mathbf{y} = \mathbf{\Phi x}  \label{eq:noisy}
\end{equation}
where $\mathbf{\Phi} \in \mathcal{R}^{m \times n}$ is called the
measurement matrix. At each iteration, OMP estimates the support
(positions of non-zero elements) of $\mathbf{x}$ by adding an index
of the column in $\mathbf{\Phi}$ that is mostly correlated with the current
residual. The vestige of columns in the estimated support is then
eliminated from the measurements $\mathbf{y}$, yielding an updated
residual for the next iteration. See~\cite{pati1993orthogonal,tropp2007signal} for details on the OMP algorithm.

While the number of iterations of the
OMP algorithm is typically set to the sparsity level $K$ of the
underlying signal to be recovered, there have been recent efforts to
relax this constraint with an aim of enhancing the recovery
performance. In one direction, an approach allowing more iterations
than the sparsity has been
suggested~\cite{livshits2012efficiency,zhang2011sparse,foucart2013stability}.
In another direction, algorithms identifying multiple indices at each
iteration have been proposed. Well-known examples include stagewise
OMP (StOMP)~\cite{donoho2006sparse}, regularized OMP
(ROMP)~\cite{needell2010signal}, CoSaMP~\cite{needell2009cosamp},
and subspace pursuit (SP)~\cite{dai2009subspace}. Key feature of
these algorithms is to introduce special operations in the identification step to select multiple promising indices. Specifically, StOMP
picks indices whose magnitudes of correlation exceed a deliberately
designed threshold. ROMP chooses a set of $K$ indices and then reduces the number of candidates using a predefined regularization rule. CoSaMP and SP add multiple indices and then prune a large portion of chosen indices to refine the identification step. 
In contrast to
these algorithms performing deliberate refinement of the identification step, a recently proposed extension of OMP, referred to as
generalized OMP (gOMP)~\cite{wang2012Generalized} (also known as
OSGA or
OMMP~\cite{liu2012orthogonal,liu2012super,huang2011recovery}),
simply chooses $S$ columns that are mostly correlated with the
residual. A detailed description of the gOMP algorithm is given in Table~\ref{gOMP}. 
%
%
%
%
%
\setlength{\arrayrulewidth}{1.5pt}
\begin{table}[t]
  \centering
      \caption[]{The gOMP Algorithm \label{gOMP}}
      \vspace{-1mm}
\begin{tabular}{@{}ll}
\hline \\ \vspace{-15pt} \\
\textbf{Input}:      &measurement matrix $\mathbf{\Phi} \in \mathcal{R}^{m \times n}$,\\
                     &measurements $\mathbf{y} \in \mathcal{R}^{m}$, \\
                     &sparsity level $K$, \\
                     &number of indices for each selection $S \leq K$. \\
\textbf{Initialize}: &iteration count $k = 0$, \\
                     &estimated list $T^{0} = \emptyset$, \\
                     &residual vector $\mathbf{r}^{0} = \mathbf{y}$.
                     \\
\textbf{While}  &$\| \mathbf{r}^{k} \|_2 > \epsilon$ and $k < \frac{m}{S}$ \textbf{do} \\
&$k = k + 1$.\\
            & $\Lambda^k = \arg \max_{\Lambda: |\Lambda| = S} \|(\mathbf{\Phi}' \mathbf{r}^{k-1})_{\Lambda}\|_1$. \hspace{0.016in}\textbf{(Identification)} \\
            & $T^k = T^{k-1} \cup \Lambda^k$. \hspace{0.993in}\textbf{(Augmentation)} \\
            & $\hat{\mathbf{x}}^k = \arg \min_{supp(\mathbf{u}) = T^k} \|\mathbf{y}-\mathbf{\Phi} \mathbf{u}\|_2$. \hspace{0.016in}\textbf{(Estimation)} \\
            & $\mathbf{r}^k = \mathbf{y} - \mathbf{\Phi} \hat{\mathbf{x}}^k$. \hspace{1.09in} \textbf{(Residual Update)} \\
\textbf{End}  \\
\textbf{Output}     & the estimated support $\hat{{T}} =
\underset{\mathcal{A}:|\mathcal{A}| = K}{\arg \min} \|\hat{\mathbf{x}}^k -
\hat{\mathbf{x}}^k_\mathcal{A}\|_2$ and  \\
& signal $\hat{\mathbf{x}}$ satisfying $\hat{\mathbf{x}}_{\{1, \cdots, n\} \setminus \hat{{T}}} = \mathbf{0}$ and $\hat{\mathbf{x}}_{\hat{{T}}} = \mathbf{\Phi}^\dag_{\hat{{T}}} \mathbf{y}$.  \\ 
\vspace{-6pt} \\
\hline
\end{tabular}
\end{table}
\setlength{\arrayrulewidth}{1.3pt}

The main motivation of gOMP is to reduce the computational
complexity and at the same time speed up the processing time. Since the OMP algorithm chooses one index is at a time, the running time and the computational complexity depend heavily on the sparsity $K$. When $K$ is large, therefore, the computational cost and the running time of OMP might be problematic. 
In the gOMP algorithm, however, more than one ``correct'' index can be chosen in each iteration due to the identification of multiple indices at a time, so that the number of iterations needed to finish the algorithm is usually much smaller than that of the OMP algorithm. 
In fact, it has been shown that the
computational complexity of gOMP is $2s m n + (2 S^2 + S) s^2 m$ where $s$ is the number of actually performed iterations~\cite{wang2012Generalized}, while that of OMP is $2K mn + 3K^2 m$. Since $s$ is in general much smaller than $K$, the gOMP algorithm runs faster than the OMP algorithm and also has lower computational complexity.

In analyzing the theoretical performance of gOMP, the restricted isometry property (RIP) has been popularly used~\cite{liu2012orthogonal,wang2012Generalized,liu2012super,huang2011recovery,maleh2011improved,satpathi2013improving,shen2014analysis,dan2014analysis,li2015sufficient}.
A measurement matrix $\mathbf{\Phi}$ is said to satisfy the RIP of
order $K$ if there exists a constant $\delta(\mathbf{\Phi}) \in
[0,1)$ such that~\cite{candes2005decoding}
\begin{equation}
    (1-\delta(\mathbf{\Phi})) \| \mathbf{x} \|_2^2 \leq \| \mathbf{\Phi}\mathbf{x} \|_2^2 \leq (1+\delta(\mathbf{\Phi})) \| \mathbf{x} \|_2^2        \label{eq:rip}
\end{equation}
for any $K$-sparse vector $\mathbf{x}$.
In particular, the minimum of all constants $\delta (\mathbf{\Phi})$
satisfying \eqref{eq:rip} is called the restricted isometry constant (RIC) and denoted as
$\delta_K (\mathbf{\Phi})$. In the sequel, we use $\delta_K$ instead
of $\delta_K (\mathbf{\Phi})$ for notational simplicity.
In~\cite{wang2012Generalized}, it has been shown that gOMP ensures the perfect recovery of any $K$-sparse signal $\mathbf{x}$ from $\mathbf{y} = \mathbf{\Phi x}$ within $K$ iterations under
\begin{equation} \label{eq:gOMPC}
 \delta_{SK} < \frac{\sqrt S}{\sqrt K + 3 \sqrt S},
\end{equation}
where $S$ is the number of indices chosen in each iteration.
In the noisy scenario where the the measurements are corrupted by a noise vector
$\mathbf{v}$ (i.e., $\mathbf{y} = \mathbf{\Phi x} + \mathbf{v}$), it has been shown that the output $\hat{\mathbf{x}} $
of gOMP after $K$ iterations satisfies~\cite{wang2012Generalized}
\begin{equation} \label{eq:noisy6}
  \|\hat{\mathbf{x}} - \mathbf{x}\|_2 \leq C \sqrt K \|\mathbf{v}\|_2
\end{equation}
under \eqref{eq:gOMPC} and $\delta_{SK + S} < 1$ where $C$ is a constant.

%
%
While the empirical recovery performance of gOMP is promising, theoretical results we just described are relatively weak when compared to the state-of-the-art recovery algorithms. For example, performance guarantees of basis pursuit (BP)~\cite{chen1998atomic} and CoSaMP are given by $\delta_{2K} < \sqrt 2 - 1$~\cite{candes2008restricted} and $\delta_{4K} < 0.1$~ \cite{needell2009cosamp}, while conditions for gOMP require that the RIC should be inversely proportional to $\sqrt K$~\cite{liu2012orthogonal,wang2012Generalized,liu2012super,huang2011recovery,maleh2011improved,satpathi2013improving,shen2014analysis,dan2014analysis,li2015sufficient}. 
Another weakness among existing theoretical results of gOMP lies in the lack of stability guarantees for the  noisy scenario where the measurements are corrupted by the noise. 
For example, it can be seen from~\eqref{eq:noisy6} that the $\ell_2$-norm of the recovery error of gOMP is upper bounded by $C \sqrt{{K}} \|\mathbf{v}\|_2$. This implies that even for a small $\|\mathbf{v}\|_2$, the $\ell_2$-norm of recovery error can be unduly large when the sparsity $K$ approaches infinity. In contrast, BP denoising (BPDN) and CoSaMP are known to have recovery error bound directly proportional to $\|\mathbf{v}\|_2$ and hence are stable under measurement noise~\cite{needell2009cosamp,chen1998atomic}.

The main purpose of this paper is to provide an improved performance analysis of the gOMP algorithm. Specifically, we show that if the measurement matrix $\mathbf{\Phi}$
satisfies the RIP with
\begin{equation} \label{eq:nonoise}
 \delta_{\max \{9, S + 1\}K} \leq \frac{1}{8},
\end{equation}
gOMP achieves stable recovery of any $K$-sparse signal $\mathbf{x}$
from the noisy measurements $\mathbf{y} = \mathbf{\Phi x} +
\mathbf{v}$ within $\max \left\{K, \left\lfloor \frac{8K}{S} \right\rfloor \right\}$ iterations. That is, the $\ell_2$-norm of recovery error satisfies
\begin{equation} \label{eq:robustbound}
  \|\hat{\mathbf{x}} - \mathbf{x}\|_2 \leq C \|\mathbf{v}\|_2,
\end{equation}
where $C$ is a constant. In the special case where $\|\mathbf{v}\|_2 = 0$ (i.e., the noise-free case), we show that gOMP accurately recovers all $K$-sparse signals in $\max \left\{K, \left\lfloor \frac{8K}{S} \right\rfloor\right\}$ iterations under 
\begin{equation}
\delta_{7K} \leq \frac{1}{8}.
\end{equation}
When compared to previous results~\cite{liu2012orthogonal,wang2012Generalized,liu2012super,huang2011recovery,maleh2011improved,satpathi2013improving,shen2014analysis,dan2014analysis,li2015sufficient},
our new results are important in two aspects. 

\vspace{1mm}
\begin{enumerate}[i)]
\item 
Our results show that the gOMP algorithm can recover sparse signals under the similar RIP condition that the state-of-the-art sparse recovery algorithms (e.g., BP and CoSaMP) require. For Gaussian random measurement matrices, this implies that the number of measurements required for gOMP is essentially $m = \mathcal{O}\left(K \log \frac{n}{K}
\right)$~\cite{candes2005decoding,baraniuk2008simple}, which is significantly smaller than the result $m = \mathcal{O}\left(K^2 \log \frac{n}{K} \right)$ obtained in previous works.


\vspace{1mm}
\item While previous work showed that the $\ell_2$-norm of the recovery error in the noisy scenario depends
linearly on $\sqrt{K} \|\mathbf{v}\|_2$~\cite{wang2012Generalized},
our new result suggests that the recovery distortion of gOMP  is upper bounded by a constant times $\|\mathbf{v}\|_2$, which strictly ensures the stability of gOMP under measurement noise.
 
\end{enumerate}

We briefly summarize notations used in this paper.
For a vector $\mathbf{x} \in \mathcal{R}^n$, $T = supp(\mathbf{x})=
\{i |x_i\neq 0\}$ represents the set of its non-zero positions.
$\Omega = \{1, \cdots, n\}$. For a set $A \subseteq \Omega$, $|A|$
denotes the cardinality of $A$. $T \backslash A$ is the set of all
elements contained in $T$ but not in $A$.
${{\mathbf{\Phi }}_A} \in {\mathcal{R}^{m \times \left| A \right|}}$
is the submatrix of ${\mathbf{\Phi }}$ that only contains columns
indexed by $A$. $\mathbf{\Phi}_{A}'$ means the transpose of the
matrix $\mathbf{\Phi}_{A}$.
${{\mathbf{x}}_A} \in \mathcal{R}^{|A|}$ is the vector which equals
$\mathbf{x}$ for elements indexed by $A$. If $\mathbf{\Phi}_A$ is
full column rank, then $\mathbf{\Phi}_A^\dagger =
(\mathbf{\Phi}'_A\mathbf{\Phi}_A)^{-1}\mathbf{\Phi}'_A$ is the
pseudoinverse of $\mathbf{\Phi}_A$.
$span(\mathbf{\Phi}_A)$ stands for the span of columns in
$\mathbf{\Phi}_A$. $\mathcal{P}_{A}=\mathbf{\Phi}_{A}
\mathbf{\Phi}_{A}^\dagger$ is the projection onto
$span(\mathbf{\Phi}_{A})$.
$\mathcal{P}_{A}^\bot = \mathbf{I}-\mathcal{P}_{A}$ is the
projection onto the orthogonal complement of
$span(\mathbf{\Phi}_{A})$. At the $k$th iteration of gOMP, we use $T^k$, $\Gamma^k = T \backslash T^k$,
$\hat{\mathbf{x}}^{k}$ and $\mathbf{r}^{k}$ to denote the estimated
support, the remaining support set, the estimated sparse signal, and
the residual vector, respectively.

%
%
%
%
%
%
%
%
%

The remainder of the paper is organized as follows. In Section
\ref{sec:thmsim}, we provide theoretical and empirical results of gOMP. In Section~\ref{sec:proof}, we
present the proof of theoretical results and conclude the
paper in Section \ref{sec:conclusion}.

%
%
\section{Sparse Recovery with gOMP} \label{sec:thmsim}

\subsection{Main Results} \label{sec:thm}

In this section, we provide performance guarantees of gOMP in recovering
sparse signals in the presence of noise. Since the noise-free scenario can be considered as a special case of the noisy scenario, extension to the noise-free scenario is straightforward.
%
%
%
%
%
%
%
%
%
%
%
%
%
%
%
%
%
%
%
%
%
%
%
%
%
%
%
%
%
%
In the noisy scenario, the perfect reconstruction of sparse signals cannot be done and hence we use the
$\ell_2$-norm of the recovery error as a performance measure. 
We first show that after a specified number of iterations, the $\ell_2$-norm of the residual of gOMP is upper bounded by a quantity depending only on $\|\mathbf{v}\|_2$.

\vspace{1mm}

\begin{theorem} \label{thm:general_2}
Let $\mathbf{x} \in \mathcal{R}^n$ be any $K$-sparse vector, $\mathbf{\Phi} \in \mathcal{R}^{m \times n}$ be a
measurement matrix, and $\mathbf{y} = \mathbf{\Phi x} + \mathbf{v}$
be the noisy measurements where $\mathbf{v}$ is a noise vector. Then under 
\begin{equation}
\delta_{\max \left\{Sk + 7 |\Gamma^k|, {Sk + S + |\Gamma^k|} \right\} } \leq \frac{1}{8}, \label{eq:RICC}
\end{equation}the residual of gOMP satisfies
\begin{equation}
  \big\|\mathbf{r}^{k + \max \big\{|\Gamma^k|, \big\lfloor \frac{8 |\Gamma^k|}{S}\big\rfloor \big\}} \big\|_2 \leq \mu_k \|\mathbf{v}\|_2
  \label{eq:result}
\end{equation} 
where $\mu_k$ is a constant depending only on $\delta_{\max \left\{Sk + 7 |\Gamma^k|, {Sk + S + |\Gamma^k|} \right\} }$.
\end{theorem}

\vspace{1mm}

The proof will be given in Section \ref{sec:proof}. One can observe from Theorem~\ref{thm:general_2} that if gOMP already performs $k$ iterations, then it requires at most $\max \left\{|\Gamma^k|, \left\lfloor  \frac{8 |\Gamma^k|}{S}\right\rfloor \right\}$ additional iterations to ensure that the $\ell_2$-norm of residual falls below $\mu_k\|\mathbf{v}\|_2$. In particular, when $k = 0$ (i.e., at the beginning of the iterations), $|\Gamma^k| = K$ and the RIC in \eqref{eq:RICC} is simplified to $\delta_{7K}$, and hence we have a simple interpretation of Theorem~\ref{thm:general_2}.

\vspace{1mm}

\begin{theorem} \label{cor:cor2}
Let $\mathbf{x} \in \mathcal{R}^n$ be any $K$-sparse vector,
$\mathbf{\Phi} \in \mathcal{R}^{m \times n}$ be the measurement
matrix, and $\mathbf{y} = \mathbf{\Phi x} + \mathbf{v}$ be the noisy
measurements where $\mathbf{v}$ is the noise vector. Then under $\delta_{7K} \leq \frac{1}{8}$,
the residual of gOMP satisfies
\begin{equation}
  \big\|\mathbf{r}^{\max \left\{K, \left\lfloor \frac{8 K }{S}\right\rfloor \right\}} \big\|_2 \leq \mu_0 \|\mathbf{v}\|_2,
\end{equation} 
where $\mu_0$ is a constant depending only on $\delta_{7K}$. 
\end{theorem}

 \vspace{1mm}
%
From Theorem~\ref{cor:cor2}, we also obtain the exact recovery condition of gOMP in the noise-free scenario. In fact, in the absence of noise, Theorem~\ref{cor:cor2} suggests that $\big\|\mathbf{r}^{\max \left\{K, \left\lfloor \frac{8 K }{S}\right\rfloor \right\}}\big\|_2 = 0$ under $\delta_{7K} \leq \frac{1}{8}$. Therefore, gOMP recovers any $K$-sparse signal accurately within $\max \{K, \lceil \frac{8K}{S} \rceil\}$ iterations under $\delta_{7K} \leq \frac{1}{8}$.

We next show that the $\ell_2$-norm of the recovery error of gOMP is also upper bounded by the product of a constant and $\|\mathbf{v}\|_2$.
 
 \vspace{1mm}
 
\begin{theorem}  [Stability\hspace{.38mm} under\hspace{.38mm}  Measurement\hspace{.38mm}  Perturbations] \label{cor:cor3}
Let $\mathbf{x} \in \mathcal{R}^n$ be any $K$-sparse vector,
$\mathbf{\Phi} \in \mathcal{R}^{m \times n}$ be the measurement
matrix, and $\mathbf{y} = \mathbf{\Phi x} + \mathbf{v}$ be the noisy
measurements where $\mathbf{v}$ is the noise vector. Then under $\delta_{\max \{9, S + 1\}K} \leq \frac{1}{8}$, 
gOMP satisfies
\begin{equation}
  \big\|\hat{\mathbf{x}}^{\max \left\{K, \left\lfloor \frac{8 K }{S}\right\rfloor \right\}} - \mathbf{x}\big \|_2 \leq \mu \|\mathbf{v}\|_2 \label{eq:10ooo}
\end{equation}
and  
\begin{equation}
 \|\hat{\mathbf{x}} - \mathbf{x}\|_2 \leq
  {C} \|\mathbf{v}\|_2, \label{eq:10ooo1}
\end{equation}
%
where $\mu$ and $C$ are constants depending on $\delta_{\max \{9, S + 1\}K}$. 
\end{theorem}

\begin{proof}
See Appendix \ref{app:cor2}.
\end{proof}

%

\vspace{1mm}

\begin{remark}[Comparison with previous results]
 \label{rem:1}
 
 
From Theorem~\ref{cor:cor2} and~\ref{cor:cor3}, we observe that gOMP is far more effective than what previous results tell.
Indeed, upper bounds in Theorem~\ref{cor:cor2} and~\ref{cor:cor3} are
absolute constants and independent of the sparsity $K$, while those in previous
works are inversely proportional to $\sqrt{K}$ (e.g., $\delta_{SK} <
\frac{\sqrt{S}}{(2 + \sqrt{2}) \sqrt{K}}$~\cite{liu2012orthogonal},
$\delta_{SK} < \frac{\sqrt{S}}{\sqrt{K} + 3 \sqrt{S}}$
\cite{wang2012Generalized}, and $\delta_{SK} <
\frac{\sqrt{S}}{\sqrt{K} + 2 \sqrt{S}}$
\cite{satpathi2013improving}). Clearly the upper bounds in previous works will vanish when $K$
is large.

\end{remark}

\vspace{1mm}

\begin{remark}[Number of measurements]
\label{rem:2}
It is well known
that a random measurement matrix $\mathbf{\Phi} \in \mathcal{R}^{m
\times n}$, which has independent and identically distributed (\text{i.i.d.}) entries with Gaussian distribution
$\mathcal{N}(0, \frac{1}{m})$, obeys the RIP with $\delta_K \leq
\varepsilon$ with overwhelming probability if
$m = \mathcal{O} \left(\frac{ K \log \frac{n}{K}}{\varepsilon^2}\right)$~\cite{candes2005decoding,baraniuk2008simple}.
When the recovery conditions in~\cite{liu2012orthogonal,wang2012Generalized,liu2012super,huang2011recovery,maleh2011improved,satpathi2013improving,shen2014analysis,dan2014analysis,li2015sufficient}
are used, the number of required measurements is expressed as $m =
\mathcal{O} \left( K^2 \log \frac{n}{K}\right)$. Whereas, our
new conditions require $m =
\mathcal{O} \left( K \log \frac{n}{K}\right)$, which is significantly smaller than the previous result, in particular for large $K$.
\end{remark}

\vspace{1mm}

\begin{remark}[Comparison with information-theoretic results]

\label{rem:3}

It might be worth comparing our result with information-theoretic results in~\cite{rangan2009asymptotic,guo2009single,donoho2009message}. Those results, which are obtained by a single-letter characterization in a large system limit, provides a performance limit of maximum a posteriori (MAP) and minimum mean square error (MMSE) estimation.
For Gaussian random measurements, it is known that the MMSE achieves a scaling of $m = \mathcal{O}(K)$ when the signal-to-noise ratio (SNR) is sufficiently large~\cite{rangan2009asymptotic,guo2009single}. In contrast, the gOMP algorithm requires $m = \mathcal{O} \left( K \log \frac{n}{K}\right)$, which is larger than the MMSE scaling in a factor of $\log \frac{n}{K}$. 
Whether one can remove the logarithm term in the number of measurements from conventional sparse recovery algorithms such as gOMP in non-limiting regime is an interesting open question.

\end{remark}

\vspace{1mm}

\begin{remark}[Recovery error]
\label{rem:c}
The constants $\mu_0$, $\mu$ and $C$ in Theorem~\ref{cor:cor2} and~\ref{cor:cor3} can be estimated from the RIC. For example, when $\delta_{7K} \leq \delta_{\max \{9, S + 1\}K} \leq 0.05$, we have $\mu_0 \leq 49$, $\mu \leq 52$ and $C \leq 110$. 
%
%
%
%
It might be interesting to compare the constant $C$ of gOMP with MMSE results~\cite{rangan2009asymptotic,guo2009single,donoho2009message}. Consider the scenario where $\mathbf{\Phi}$ is a random matrix having \text{i.i.d.} elements of zero mean and $\frac{1}{m}$ variance, $\mathbf{x}$ is a sparse vector with each non-zero element taking the value $\pm 1$ with equal probability, and $\mathbf{v}$ is the noise vector with \text{i.i.d.} Gaussian elements. Consider $n = 10, 000$ and $m = 500$ and suppose $\mathbf{x}$ has sparsity rate $p = 0.001$ (so that the sparsity level $K$ is $10$ on average). Then for an SNR of $0$ dB, the required bound of MMSE is $8.6 \times 10^{-6}$ (per dimension)~\cite{guo2009single}, which amounts to $\|\mathbf{x} - \hat{\mathbf{x}}\|_2 \leq 0.027 \|\mathbf{v}\|_2$.\footnote{SNR = $0$ dB implies that $\frac{\|\mathbf{\Phi x}\|^2_2}{\|\mathbf{v}\|_2^2} = 1$. Since each element in $\mathbf{\Phi}$ has power $\frac{1}{m}$, we have $\mathbb{E} \big[ (\mathbf{\Phi x})_j^2 \big] = \frac{p n}{m} = \frac{1}{50}$, which implies that $\mathbb{E} \big[ {v}^2_j \big] = \frac{1}{50}$ and hence $\mathbb{E} \big[ \|\mathbf{v}\|_2] = \sqrt{10}$. 
}
Clearly, the constant $0.027$ for the MMSE result is much smaller than the constant $C$ obtained from gOMP. In fact, the constant $C$ obtained in Theorem~\ref{cor:cor3} is generally loose, as we will see in the simulations. This is mainly because 1) our analysis is based on the RIP framework so that the analysis is in essence the worst-case-analysis, and 2) many relaxations are used in our analysis to obtain the constant upper bounds. 
  
\end{remark}

\vspace{1mm}

\begin{remark}[Comparison with OMP]

\label{rem:4}
When $S = 1$, gOMP returns to the OMP algorithm and Theorem~\ref{cor:cor3} suggests that OMP performs stable recovery of all $K$-sparse signals in $\max \left\{K, \left\lfloor \frac{8 K }{S}\right\rfloor \right\} = 8K$ iterations under $\delta_{9K} \leq \frac{1}{8}$. In a recent work of Zhang~\cite[Theorem 2.1]{zhang2011sparse}, it has been shown that OMP can achieve stable recovery of $K$-sparse signals in $16.6K$ iterations under $\delta_{17.6K} \leq \frac{1}{8}$. Clearly, our new result indicates that OMP has better (less restrictive) RIP condition and also requires smaller number of iterations.  
\end{remark}

\vspace{1mm}

It is worth mentioning that even though the input signal is not strictly sparse, in many cases, it can be well approximated by a sparse signal. Our result can be readily extended to this scenario.

%

\vspace{1mm}

\begin{corollary}[Recovery of non-sparse signals] \label{cor:nonsparse}
Let $\mathbf{x}_K \in \mathcal{R}^n$ be the vector that keeps $K$ largest elements of the input vector $\mathbf{x}$ and sets all other entries to zero. 
%
%
%
Let $\mathbf{\Phi} \in \mathcal{R}^{m \times n}$ be the measurement
matrix satisfying $\delta_{\max \{18, 2S + 2\}K} \leq \frac{1}{8}$ and $\mathbf{y} = \mathbf{\Phi x} + \mathbf{v}$ be the noisy
measurements where $\mathbf{v}$ is the noise vector. Then, 
gOMP produces an estimate $\hat{\mathbf{x}}$ of $\mathbf{x}$ in $\max \left\{2K, \left\lfloor \frac{16K}{S}
\right\rfloor \right\}$ iterations such that
\begin{eqnarray}
\|\hat{\mathbf{x}} - \mathbf{x}\|_2 \leq {D} \left(\frac{\|\mathbf{x} - \mathbf{x}_K\|_1}{\sqrt K} + \|\mathbf{v}\|_2 \right), 
\end{eqnarray} 
where $D$ is a constant depending on $\delta_{\max \{18, 2S + 2\}K}$.
\end{corollary}

\vspace{1mm}
%
%
Since Corollary \ref{cor:nonsparse} is a straightforward extension of Theorem \ref{cor:cor3}, we omit the proof for brevity (see~\cite{needell2010signal,blumensath2009iterative,foucart2011hard,needell2009cosamp,dai2009subspace}).
Note that the key idea is to partition the noisy measurements of a non-sparse signal into two parts and then apply Theorem~\ref{cor:cor3}. The two parts consist of 1) measurements associated with dominant elements of the signal ($\mathbf{y}_1 = \mathbf{\Phi x}_K$) and 2) measurements associated with insignificant elements and the noise vector ($\mathbf{y}_2 = \mathbf{\Phi} (\mathbf{x} - \mathbf{x}_K) + \mathbf{v}$).  
That is, 
\begin{equation}
\mathbf{y} = \mathbf{y}_1 + \mathbf{y}_2 = \mathbf{\Phi} \mathbf{x}_K + \mathbf{\Phi} (\mathbf{x} - \mathbf{x}_K) + \mathbf{v}.
\end{equation} 

\subsection{Empirical Results} \label{sec:sim}
We evaluate the recovery performance of the gOMP algorithms through numerical experiments. Our simulations are focused on the noisy scenario (readers are referred to~\cite{maleh2011improved,huang2011recovery,wang2012Generalized} for simulation results in the noise-free scenario).  
In our simulations, we consider random matrices $\mathbf{\Phi}$ of size $100 \times 200$ whose entries are drawn \text{i.i.d.} from Gaussian distribution $\mathcal{N} (0, \frac{1}{m})$. We generate $K$-sparse signals $\mathbf{x}$ whose components are \text{i.i.d.} and follow a Gaussian-Bernoulli distribution
\begin{equation}
 {x}_j \sim   \begin{cases}
0 & \text{with probability}~1 - p, \\
\mathcal{N} (0, 1) &\text{with probability}~p,
\end{cases} 
\end{equation}
where $p$ is the sparsity rate that represents the average fraction of non-zero components in $\mathbf{x}$.
We employ the mean square error (MSE) as a metric to evaluate the recovery performance in the noisy scenario. The MSE is defined as
\begin{equation}
\text{MSE}= \frac{1}{n} \sum_{i = 1}^n ({x}_i - \hat{{x}}_i)^2,
\end{equation}
where $\hat{{x}}_i$ is the estimate of $x_i$. 
In our simulation, the following recovery algorithms are considered:
\begin{enumerate}
\item OMP and gOMP ($S = 3, 5$).

\item CoSaMP: We set the maximal iteration number to $50$ to avoid repeated iterations (\url{http://www.cmc.edu/pages/faculty/DNeedell})..

\item StOMP: We use false alarm rate control strategy as it works better than false discovery rate control strategy (\url{http://sparselab.stanford.edu/}).


\item 
BPDN (\url{http://cvxr.com/cvx/}).

\item Generalized approximate message passing (GAMP)~\cite{donoho2009message,rangan2011generalized,vila2011expectation}: (\url{http://gampmatlab.wikia.com/}).

\item Linear MMSE estimator.
\end{enumerate}
In obtaining the performance result for each simulation point of the algorithm, we perform $2,000$ independent trials.
  

\begin{figure}[t] 
\centering 
\subfigure[$p = 0.05$.] 
{\hspace{-2mm} 
\includegraphics[width = 90 mm] {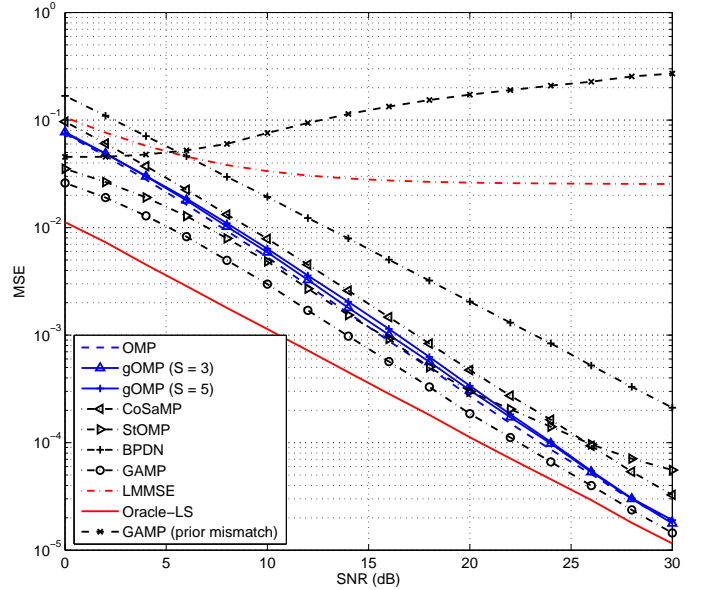}
\label{fig:subfig050}} 
\subfigure[$p = 0.1$.] 
{\hspace{-2mm} 
\includegraphics[width = 90 mm] {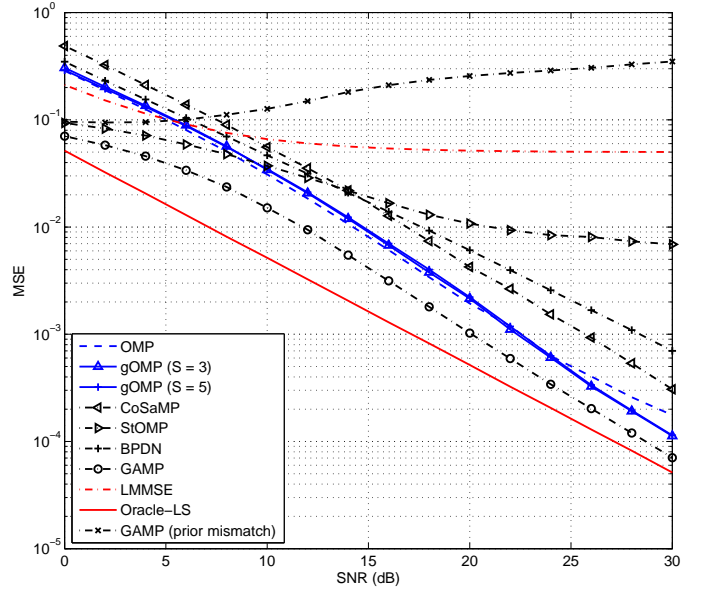}
\label{fig:subfig100}} 
\caption{MSE performance of recovery algorithms as a function of SNR.}
\label{fig:mse1}
\end{figure}

 
In Fig.~\ref{fig:mse1}, we plot the MSE performance for each recovery method as a function of signal-to-noise ratio (SNR), where the SNR (in dB) is defined as 
\begin{equation}
\text{SNR} = 10 \log_{10} \frac{\|\mathbf{\Phi x}\|_2^2}{\|\mathbf{v}\|_2^2}.
\end{equation}
In this case, the system model is expressed as $\mathbf{y} = \mathbf{\Phi x} + \mathbf{v}$ where $\mathbf{v}$ is the noise vector whose elements are generated from Gaussian distribution $\mathcal{N}(0, \frac{p n}{m} 10^{- \frac{\textbf{SNR}}{10}})$.\footnote{Since the components of $\mathbf{\Phi}$ have power $\frac{1}{m}$ and the signal $\mathbf{x}$ has sparsity rate $p$, $\mathbb{E}|(\mathbf{\Phi x})_i|^2 = \frac{p n}{m}$. From the definition of SNR, we have $\mathbb{E}|v_i|^2 = \mathbb{E}|(\mathbf{\Phi x})_i|^2 \cdot 10^{- \frac{\textbf{SNR}}{10}} = \frac{p n}{m} 10^{- \frac{\textbf{SNR}}{10}}$.} The benchmark performance of Oracle least squares estimator  (Oracle-LS), the best possible estimation having prior knowledge on the support of input signals, is plotted as well. In general, we observe that for all methods, the MSE performance improves with the SNR. While GAMP has the lowest MSE when the prior knowledge on the signal and noise distribution is available, it does not perform well when the prior information is incorrect.\footnote{In order to test the mismatch scenario, we use Bernoulli distribution (${x}_j \sim \mathcal{B}(1, p)$).}
For the whole SNR region under test, the MSE performance of gOMP is comparable to OMP and also outperforms CoSaMP and BPDN. 
An interesting point is that the actual recovery error of gOMP is much smaller than that provided in Theorem~\ref{cor:cor3}. For example, when SNR = $10$ dB, $p = 0.05$, and ${v}_j \sim \mathcal{N}(0, \frac{p n}{m} 10^{- \frac{\textbf{SNR}}{10}})$, we have $\mathbb{E} \|\mathbf{v}\|_2 = ({p n} 10^{- \frac{\textbf{SNR}}{10}})^{1/2} = 1$. Using this together with~\eqref{eq:Cvalue}, the upper bound for $\|\mathbf{x} - \hat{\mathbf{x}}\|_2$ in Theorem~\ref{cor:cor3} is around $63$.
In contrast, when SNR = $10$ dB, the $\ell_2$-norm of the actual recovery error of gOMP is $\|\mathbf{x} - \hat{\mathbf{x}}\|_2 = (n \cdot \text{MSE})^{1/2} \approx 1$ (Fig.~\ref{fig:subfig050}), which is much smaller than the upper bound indicated in Theorem~\ref{cor:cor3}.

%
%
%

\begin{figure}[t] 
\centering
\hspace{-2mm}{\includegraphics[width = 90 mm]{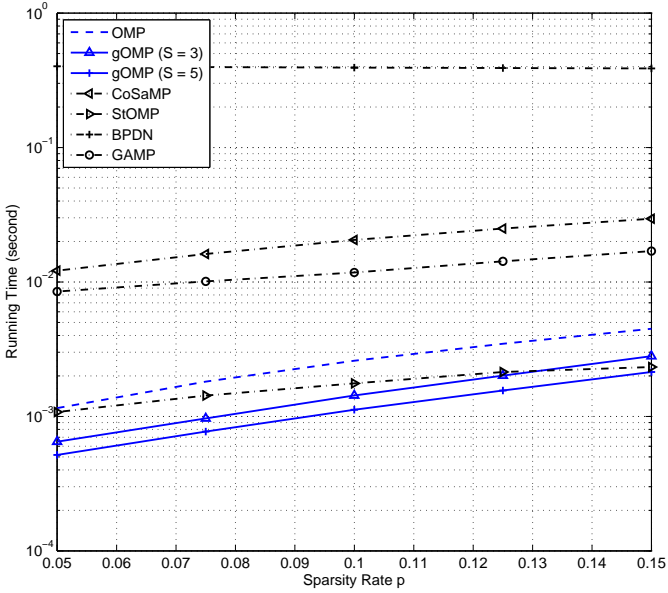}} 
\caption{Running time as a function of sparsity rate $p$.} \label{fig:time} 
\end{figure}


Fig.~\ref{fig:time} displays the running time of each recovery method as a function of the sparsity rate $p$. The running time is measured using the MATLAB program on a personal computer with Intel Core i7 processor and Microsoft Windows $7$ environment. Overall, we observe that the running time of OMP, gOMP, and StOMP is smaller than that of CoSaMP, GAMP, and BPDN. In particular, the running time of BPDN is more than one order of magnitude higher than the rest of algorithms require. This is because the complexity of BPDN is a quadratic function of the number of measurements ($\mathcal{O}(m^2 n^{3/2})$)~\cite{nesterov1994interior}, while that of the gOMP algorithm is $\mathcal{O}(Kmn)$~\cite{wang2012Generalized}.
Since  gOMP can chooses more than one support index at a time, we also observe that gOMP runs faster than the OMP algorithm.

\section{Proof of Theorem~\ref{thm:general_2}} \label{sec:proof}

%
%
\subsection{Preliminaries} \label{sec:preliminaries} Before we
proceed to the proof of Theorem~\ref{thm:general_2}, we present definitions used in our analysis.
Recall that $\Gamma^k = T \backslash T^k$ is the set of
remaining support elements after $k$ iterations of gOMP. In what
follows, we assume without loss of generality that $\Gamma^{k} =
\left\{1, \cdots,|\Gamma^{k}|\right\}$. Then it is clear that $0 \leq
|\Gamma^{k}| \leq K$. For example, if $k = 0$, then $T^{k} =
\emptyset$ and $|\Gamma^{k}| = |T| = K$. Whereas if $T^{k} \supseteq
T$, then $\Gamma^{k} = \emptyset$ and $|\Gamma^{k}| = 0$. Also, for
notational convenience we assume that $\{x_i\}$ is arranged in
descending order of their magnitudes, i.e., $ |x_1| \geq |x_2| \geq \cdots
\geq |x_{|\Gamma^{k}|}|$. Now, we define the subset
${\Gamma}^k_{\tau}$ of ${\Gamma^{k}}$ as (see Fig.
\ref{fig:subfig21}):
\begin{equation}
  {\Gamma}^k_{\tau} =
  \begin{cases}
  \emptyset &\tau = 0, \\
  \left\{1, \cdots, 2^{\tau - 1} S \right\} &\tau = 1, \cdots, \max \left\{0, \left\lceil \log_2 \frac{{|\Gamma^{k}|}}{S} \right\rceil  \right\}, \\
  \Gamma^{k} &\tau = \max \left\{0, \left\lceil \log_2 \frac{|\Gamma^{k}|}{S} \right\rceil \right\} + 1.
\end{cases} \label{eq:jjjjffff}
\end{equation}
Note that the last set ${\Gamma}^k_{\max \left\{0, \left\lceil \log_2 \frac{|\Gamma^{k}|}{S} \right\rceil \right\} + 1}$ ($= \Gamma^k$) does not necessarily have
$2^{\max \left\{0, \left\lceil \log_2 \frac{|\Gamma^{k}|}{S} \right\rceil \right\}} S$ elements.

\begin{figure}[t] \label{fig:set}
\centering \subfigure[Set diagram of $T$, $T^{k}$, and
$\Gamma^{k}$.] {\includegraphics[scale = 1] {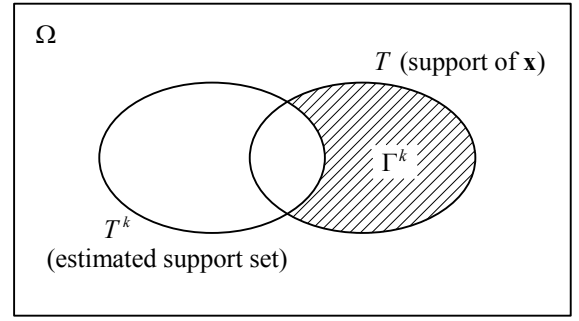}
\label{fig:subfig11}} \subfigure[Illustration of indices in
$\Gamma^{k}_\tau$.] {\includegraphics[scale = 1] {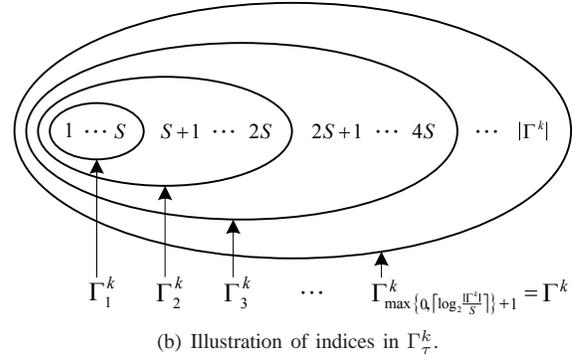}
\label{fig:subfig21}} \caption{Illustration of sets $T$, $T^{k}$,
and $\Gamma^{k}$.}
\end{figure}

For given set $\Gamma^{k}$ and constant $\sigma \geq 2$, let $L
\in \left\{1, 2, \cdots, \max \left\{0, \left\lceil \log_2 \frac{|\Gamma^{k}|}{S} \right\rceil \right\} + 1 \right\}$ be a positive integer satisfying\footnote{We note that $L$ is a
function of $k$.} 
\begin{subequations}
\begin{align}
 \|\mathbf{x}_{{\Gamma^{k}} \backslash {\Gamma}^k_{0}}\|_2^2 &< \sigma
 \|\mathbf{x}_{{\Gamma^{k}} \backslash {\Gamma}^k_{1}}\|_2^2,
 \label{eq:mu10}\\
 \|\mathbf{x}_{{\Gamma^{k}} \backslash {\Gamma}^k_{1}} \|_2^2 &< \sigma
 \|\mathbf{x}_{{\Gamma^{k}} \backslash {\Gamma}^k_{2}}\|_2^2,
 \label{eq:mu2} \\
 &~~\vdots \nonumber \\
 \|\mathbf{x}_{{\Gamma^{k}} \backslash {\Gamma}^k_{{L}- 2}}\|_2^2 &< \sigma
 \|\mathbf{x}_{{\Gamma^{k}} \backslash {\Gamma}^k_{{L}- 1}}\|_2^2,
 \label{eq:mu1} \\
 \|\mathbf{x}_{{\Gamma^{k}} \backslash {\Gamma}^k_{{L}- 1}}\|_2^2 &\geq \sigma
 \|\mathbf{x}_{{\Gamma^{k}} \backslash {\Gamma}^k_{L}}\|_2^2.
 \label{eq:mu4}
\end{align}
\end{subequations}
If \eqref{eq:mu4} holds true for all $L \geq 1$, then we ignore \eqref{eq:mu10}--\eqref{eq:mu1} and simply take $L = 1$. Note that $L$ always exists because $
\big\|\mathbf{x}_{\Gamma^{k} \backslash \Gamma^k_{\max \left\{0, \lceil \log_2 {{|\Gamma^{k}|}/{S}} \rceil \right\} + 1}}\big\|_2^2 = 0$ so that \eqref{eq:mu4}   holds true at least for $L = {\max \left\{0, \lceil \log_2 {{|\Gamma^{k}|}/{S}} \rceil \right\} + 1}$.
%
From \eqref{eq:mu10}--\eqref{eq:mu4}, we have 
\begin{equation}  \label{eq:next22}
  \| \mathbf{x}_{{\Gamma^{k}} \backslash {\Gamma}^k_\tau} \|_2^2 \leq \sigma^{{L}- 1 - \tau} \|
\mathbf{x}_{{\Gamma^{k}} \backslash {\Gamma}^k_{{L}- 1}}
  \|_2^2,~~\tau = 0, 1, \cdots, {L}.
\end{equation}
Moreover, if $L \geq 2$, we have a lower bound for $|{\Gamma}^k|$ as (see
Appendix~\ref{app:relations}):
\begin{equation} |\Gamma^{k}|
> \left(\frac{2\sigma - 1}{2\sigma  - 2}\right) 2^{{L}- 2} S.
\label{eq:next2}
\end{equation}
Equations \eqref{eq:next22} and \eqref{eq:next2} will be used in the proof of Theorem~\ref{thm:general_2} and  we will fix $\sigma = \frac{1}{2}\exp \left(\frac{14}{9}\right)$ in the proof. 

We provide two propositions useful in the proof of Theorem~\ref{thm:general_2}. The first one offers an upper bound for $\|\mathbf{r}^k\|_2^2$ and a lower bound for $\|\mathbf{r}^l \|_2^2 - \| \mathbf{r}^{l + 1} \|_2^2$ ($l \geq k$).

\vspace{1mm}

\begin{proposition} \label{prop:rk}
For given $\Gamma^{k}$ and any integer $l \geq k$, the residual of gOMP
satisfies
\begin{eqnarray} \label{eq:residual3}
&& \| \mathbf{r}^{k} \|_2^2 \leq \|\mathbf{\Phi}_{\Gamma^k}\mathbf{x}_{\Gamma^{k}} +
\mathbf{v}\|_2^2,  \\
&& \|\mathbf{r}^l\|_2^2 - \|\mathbf{r}^{l + 1}\|_2^2 \geq 
\frac{1 - \delta_{|{\Gamma}^k_{\tau} \cup T^l |}} {(1 + \delta_{S})
\left\lceil \frac{|{\Gamma}^k_\tau |}{S} \right\rceil} \nonumber \\
 &&\hspace{10mm} \times \left( \|\mathbf{r}^l \|_2^2 -
\|\mathbf{\Phi}_{{\Gamma^{k}} \backslash {\Gamma}^k_\tau}
\mathbf{x}_{{\Gamma^{k}} \backslash {\Gamma}^k_\tau} + \mathbf{v}
\|_2^2 \right),  \label{eq:residual31}
\end{eqnarray}
where $\tau = 1, 2, \cdots, \max \left\{0, \left\lceil \log_2 {\frac{|\Gamma^{k}|}{S}}
\right\rceil \right\} + 1$.
 
\end{proposition}

\begin{proof}
See Appendix \ref{app:prop1}.   
\end{proof}

\vspace{1mm}

The second proposition is essentially an extension of \eqref{eq:residual31}. It characterizes the relationship between residuals
of gOMP in different number of iterations.

\vspace{1mm}

\begin{proposition} \label{prop:residual}
For any integer $l \geq k$, $\Delta l > 0$, and $\tau \in  \{1,
\cdots, \max \left\{0, \left\lceil \log_2 \frac{|\Gamma^{k}|}{S} \right\rceil \right\} + 1\}$, the
residual $\mathbf{r}^{l + \Delta l}$ of gOMP satisfies
\begin{eqnarray}
\lefteqn{\| \mathbf{r}^{l + \Delta l} \|_2^2 -
\|\mathbf{\Phi}_{{\Gamma^{k}} \backslash {\Gamma}^k_\tau}
\mathbf{x}_{{\Gamma^{k}} \backslash {\Gamma}^k_\tau} +
\mathbf{v}\|_2^2} \nonumber \\ &&\leq C_{\tau,l, \Delta l} \left(\|
\mathbf{r}^l \|_2^2 - \|\mathbf{\Phi}_{{\Gamma^{k}} \backslash
{\Gamma}^k_\tau} \mathbf{x}_{{\Gamma^{k}} \backslash
{\Gamma}^k_\tau} + \mathbf{v}\|_2^2\right), 
  \label{eq:good1}
\end{eqnarray}
where
\begin{equation}
C_{\tau,l, \Delta l} = \exp \left( - \frac{\Delta l (1 - \delta_{|
{\Gamma}^k_\tau \cup T^{l + \Delta l - 1}|})} {\left\lceil
\frac{|{\Gamma}^k_\tau |}{S} \right\rceil (1 + \delta_{S})} \right).
\label{eq:constantc}
\end{equation}
\end{proposition}

\begin{proof}
See Appendix \ref{app:prop2}.   \end{proof}

\subsection{Outline of Proof}  \label{sec:outline}
The proof of Theorem \ref{thm:general_2} is based on mathematical induction in $|\Gamma^{k}|$, the number of remaining
indices after $k$ iterations of gOMP. We first consider the
case when $|\Gamma^{k}| = 0$. This case is trivial since all support
indices are already selected ($T \subseteq T^k$) and hence
\begin{eqnarray}
  \|\mathbf{r}^k\|_2 &=& \|\mathbf{y} - \mathbf{\Phi} \hat{\mathbf{x}}^k\|_2 \nonumber \\
  &=& \min_{supp(\mathbf{u}) = T^k} \|\mathbf{y}-\mathbf{\Phi} \mathbf{u}\|_2 \nonumber \\
  &\leq& \|\mathbf{y} - \mathbf{\Phi} \mathbf{x}\|_2 \nonumber \\
  &=& \|\mathbf{v}\|_2 \nonumber \\
  &\leq& \mu_k \|\mathbf{v}\|_2.  \label{eq:rabc}
\end{eqnarray}

Next, we assume that the argument holds up to an integer $\gamma -
1$. Under this inductive assumption, we will prove that it also holds true for $|\Gamma^{k}| = \gamma$. In other words, we will show that when $|\Gamma^{k}| = \gamma$,
\begin{eqnarray} 
\big\|\mathbf{r}^{k + \max \left\{\gamma, \left\lfloor \frac{8
\gamma}{S}\right\rfloor \right\}} \big\|_2 \leq \mu_k \|\mathbf{v}\|_2 \label{eq:42}
\end{eqnarray}
holds true under 
\begin{equation}
\delta_{\max \left\{Sk + 7 \gamma,  Sk + S + \gamma \right\}} \leq \frac{1}{8}. \label{eq:RIC00}
\end{equation}

Although the details of the proof in the induction step are somewhat cumbersome, the main idea is rather simple. 
First, we show that a decent amount of support indices in $\Gamma^k$ can be selected within a specified number of additional iterations so that the number of remaining support indices is upper bounded. More precisely,

\begin{enumerate}[i)]
\item If $L = 1$, the number of remaining support indices after $(k + 1)$ iterations is upper bounded as
\begin{equation}  \label{eq:jjff00}
|\Gamma^{k + 1}| < \gamma. 
\end{equation}

\item 
If $L \geq 2$, the number of remaining support indices after $(k
+ k_L)$ iterations satisfies
\begin{equation} \label{eq:jjff0}
 |\Gamma^{k + k_{L}} | < |\Gamma^k \backslash \Gamma^k_{L - 1}|,
\end{equation}
where
\begin{equation} \label{eq:ki'}
   k_i = 2 \sum_{\tau = 0}^i  \left \lceil\frac{|{\Gamma}^k_{\tau}|}{S} \right \rceil,~~ i = 0, \cdots,   {L}.
\end{equation}

\end{enumerate} 

Second, since \eqref{eq:jjff00} and \eqref{eq:jjff0} imply that the number of remaining support indices is no more than $\gamma - 1$, from the induction hypothesis we have
\begin{eqnarray}
  \big\|\mathbf{r}^{k + 1  +  \max \left\{\gamma, \left\lfloor \frac{8}{S}  |\Gamma^{k + 1}| \right\rfloor \right\}} \big\|_2 \leq  \mu_k
\|\mathbf{v}\|_2,\hspace{6mm}  ~L = 1, \label{eq:1111p0} \\
\big\|\mathbf{r}^{k + k_L  + \max \left\{\gamma, \left\lfloor \frac{8}{S} |\Gamma^{k + k_L}| \right\rfloor \right\}} \big\|_2 \leq  \mu_k
\|\mathbf{v}\|_2, ~~~L \geq 2.  \label{eq:1111p}
\end{eqnarray}
Further, by estimating $k + 1  +  \max \left\{\gamma, \left\lfloor \frac{8}{S}  |\Gamma^{k + 1}| \right\rfloor \right\}$ in~\eqref{eq:1111p0} and $k + k_L  + \max \left\{\gamma, \left\lfloor \frac{8}{S} |\Gamma^{k + k_L}| \right\rfloor \right\}$ in~\eqref{eq:1111p}, we establish the induction step. Specifically, 

\vspace{1mm}

\begin{enumerate}[i)]
\item $L = 1$ case: We obtain from \eqref{eq:jjff00} that
\begin{eqnarray}
\lefteqn{k + 1  +  \max \left\{\gamma, \left\lfloor \frac{8}{S}  |\Gamma^{k + 1}|  \right\rfloor \right\} }\nonumber \\
&\leq& k + 1  +  \max \left\{\gamma, \left\lfloor \frac{8}{S}  (\gamma - 1) \right\rfloor \right\} \nonumber \\
&\leq& k + \max \left\{\gamma, \left\lfloor \frac{8 \gamma}{S}   \right\rfloor \right\}.\label{eq:90}
\end{eqnarray}
By noting that the residual power of gOMP is non-increasing ($\|\mathbf{r}^i\|_2
\leq \|\mathbf{r}^j\|_2$ for $i \geq j$), we have  
\begin{eqnarray}
\big\|\mathbf{r}^{k + \max \left\{\gamma,  \left\lfloor \frac{8 \gamma}{S} \right\rfloor \right\}}  \big\|_2
    &\leq& \big\|\mathbf{r}^{k + 1  +  \max\left\{\gamma, \left\lfloor \frac{8}{S}  |\Gamma^{k + 1}| \right\rfloor \right\}} \big\|_2 \nonumber \\
    &\leq& \mu_k
\|\mathbf{v}\|_2. 
\end{eqnarray} 

\item $L \geq 2$ case: We observe from \eqref{eq:ki'} that
\begin{eqnarray}
 k_{L} &=& 2 \sum_{\tau = 0}^L  \left \lceil\frac{|{\Gamma}^k_{\tau}|}{S} \right \rceil  \nonumber \\
 &=& 2 \sum_{\tau = 1}^L  \left \lceil\frac{|{\Gamma}^k_{\tau}|}{S} \right \rceil  \nonumber \\
 &\leq& 2 \sum_{\tau = 1}^{L}  2^{\tau - 1}  \nonumber \\
 &=&  2 (2^{L} - 1), \label{eq:aaaaaaa2}
\end{eqnarray}
which together with \eqref{eq:jjff0} implies that
\begin{eqnarray}
  \lefteqn{k + k_L  + \max \left\{\gamma,  \left\lfloor \frac{8}{S} |\Gamma^{k + k_L}|
\right\rfloor \right\}} \nonumber \\
  &\leq& k + 2 ( 2^{L} - 1) + \max \left\{\gamma, \left\lfloor\frac{8}{S} |\Gamma^k \backslash \Gamma^k_{L - 1}| \right\rfloor \right\} \nonumber \\ 
  &=& k + 2 ( 2^{L} - 1) + \max \left\{\gamma, \left\lfloor\frac{8}{S} (\gamma - 2^{{L}- 2} S) 
  \right\rfloor \right\} \nonumber \\  
  &\leq& k + \max \left\{\gamma, \left\lfloor\frac{8 \gamma}{S} \right\rfloor \right\},
  \label{eq:48}
\end{eqnarray}
Hence, we obtain from \eqref{eq:1111p} and \eqref{eq:48} that 
\begin{eqnarray}
\big\|\mathbf{r}^{k + \max \left\{\gamma,  \left\lfloor \frac{8 \gamma}{S} \right\rfloor \right\}} \big\|_2
    &\leq& \big\|\mathbf{r}^{k + k_L  + \max\left\{\gamma, \left\lfloor \frac{8}{S} |\Gamma^{k + k_L}| \right\rfloor \right\}} \big\|_2 \nonumber \\
    &\leq& \mu_k
\|\mathbf{v}\|_2. \label{eq:39wuxuzheng}
\end{eqnarray} 
\end{enumerate}

In summary, what remains now is the proofs for \eqref{eq:jjff00} and~\eqref{eq:jjff0}.

\subsection{Proof of \eqref{eq:jjff0}} 
We consider the proof of \eqref{eq:jjff0} for the case of $L \geq 2$. Instead of directly proving
\eqref{eq:jjff0}, we show that a sufficient condition for
\eqref{eq:jjff0} is true. To be specific, since $\mathbf{x}_{\Gamma^{k} \backslash
{\Gamma}^k_{{L}- 1}}$ consists of $|{\Gamma^{k} \backslash
{\Gamma}^k_{{L}- 1}}|$ smallest non-zero elements (in magnitude) in $\mathbf{x}_{\Gamma^{k}}$, a sufficient condition for \eqref{eq:jjff0} is
\begin{equation} \label{eq:aasssd1}
\| \mathbf{x}_{\Gamma^{k + k_{L}}} \|_2^2 <
\|\mathbf{x}_{\Gamma^{k} \backslash {\Gamma}^k_{{L}- 1}} \|_2^2.
\end{equation}
In this subsection, we show that \eqref{eq:aasssd1} is true under \begin{equation}
\delta_{Sk + 7 \gamma } \leq \frac{1}{8}. \label{eq:cond1}
\end{equation}
To the end, we first construct lower and upper
bounds for $\| \mathbf{r}^{k + k_{L}} \|_2$ and then use these
bounds to derive a condition guaranteeing \eqref{eq:aasssd1}. 

\vspace{1mm}

\noindent \textbf{1) Lower bound for $\| \mathbf{r}^{k + k_{L}}\|_2^2$}:
\begin{eqnarray}
\lefteqn{\| \mathbf{r}^{k + k_{L}} \|_2} \nonumber \\ 
\hspace{-2mm} &=& \hspace{-2mm} \|\mathbf{y} - \mathbf{\Phi}  \hat{\mathbf{x}}^{k + k_{L}} \|_2 \nonumber \\
 \hspace{-2mm} &=& \hspace{-2mm} \| \mathbf{\Phi} ( \mathbf{x} - \hat{\mathbf{x}}^{k + k_{L}} )  + \mathbf{v} \|_2 \nonumber  \\
 \hspace{-2mm} &\geq& \hspace{-2mm} \| \mathbf{\Phi} (\mathbf{x} - \hat{\mathbf{x}}^{k + k_{L}}) \|_2 - \| \mathbf{v} \|_2 \nonumber  \\
 \hspace{-2mm} &\overset{(a)}{\geq}& \hspace{-2mm} \left(1 - \delta_{|T \cup T^{k + k_{L}}|}\right)^{1/2} \| \mathbf{x} - \hat{\mathbf{x}}^{k + k_{L}}  \|_2 - \| \mathbf{v} \|_2  \nonumber \label{eq:galdg} \\
 \hspace{-2mm} &\geq&  \hspace{-2mm}\left(1 - \delta_{|T \cup T^{k + k_{L}}|}\right)^{1/2} \| \mathbf{x}_{\Gamma^{k + k_{L}}} \|_2 - \| \mathbf{v}\|_2,  \label{eq:low}
\end{eqnarray}
where (a) is from the RIP (note that $\mathbf{x} - \hat{\mathbf{x}}^{k +
k_{L}}$ is supported on $T \cup T^{k + k_{L}}$). 

\vspace{1mm}

\noindent \textbf{2) Upper bound for $\|\mathbf{r}^{k + k_{L}}\|_2^2$}:
\vspace{1mm}

First, by applying Proposition \ref{prop:residual}, we have
\begin{eqnarray}
 \lefteqn{\hspace{-7mm}  \|{{\mathbf{r}}^{k + k_1}} \|_{2}^{2} -
\|\mathbf{\Phi}_{{\Gamma^{k}} \backslash {\Gamma}^k_1}
\mathbf{x}_{{\Gamma^{k}} \backslash {\Gamma}^k_1} + \mathbf{v}
\|_2^2 \leq C_{1,k, k_1}} \nonumber \\
&&\hspace{-18mm} \times \left(\|\mathbf{r}^{k} \|_2^2 -
\|\mathbf{\Phi}_{{\Gamma^{k}} \backslash {\Gamma}^k_1}
\mathbf{x}_{{\Gamma^{k}} \backslash {\Gamma}^k_1} +
\mathbf{v}\|_2^2\right), \label{eq:11f2}
\\
 \lefteqn{\hspace{-7mm}  \| \mathbf{r}^{k + k_2} \|_2^2 - \|\mathbf{\Phi}_{{\Gamma^{k}}
\backslash {\Gamma}^k_2} \mathbf{x}_{{\Gamma^{k}} \backslash
{\Gamma}^k_2} + \mathbf{v} \|_2^2 \leq C_{2, k + k_1, k_2 -
k_1}} \nonumber \\
&&\hspace{-18mm} \times  \left(\| \mathbf{r}^{k + k_1} \|_2^2 -
\|\mathbf{\Phi}_{{\Gamma^{k}} \backslash {\Gamma}^k_2}
\mathbf{x}_{{\Gamma^{k}} \backslash {\Gamma}^k_2} + \mathbf{v}
\|_2^2\right),  \label{eq:12f2}
\\
 &~~~~~~~~~~~~&~\vdots \nonumber
\\
 \lefteqn{\hspace{-7mm} \| \mathbf{r}^{k + k_{L}} \|_2^2 -
\|\mathbf{\Phi}_{{\Gamma^{k}} \backslash {\Gamma}^k_{L}}
\mathbf{x}_{{\Gamma^{k}} \backslash {\Gamma}^k_{L}} + \mathbf{v}
\|_2^2 \leq C_{L, k + k_{{L}- 1}, k_L - k_{L - 1}}} \nonumber \\
&&\hspace{-18mm} \times     \left(\| \mathbf{r}^{k + k_{{L}- 1}} \|_2^2 -
\|\mathbf{\Phi}_{{\Gamma^{k}} \backslash {\Gamma}^k_{L}}
\mathbf{x}_{{\Gamma^{k}} \backslash
{\Gamma}^k_{L}}  + \mathbf{v} \|_2^2\right). 
\label{eq:13f2}
\end{eqnarray}
From \eqref{eq:ki'} and monotonicity of the RIC, we have 
\begin{eqnarray}
 C_{i, k + k_{i - 1}, k_i - k_{i - 1}} &=& \exp \left( - 2 \cdot \frac{1 - \delta_{|{\Gamma}^k_i \cup T^{k + k_{i} - 1}|} }{1 + \delta_{S}}
 \right) \nonumber \\
 &\overset{(a)}{\leq}& \exp \left(- 2 \cdot \frac{1 - \delta_{Sk + 7 \gamma}}{1 + \delta_{Sk + 7 \gamma}} \right)  \nonumber \\
  &\overset{(b)}{\leq}& \exp\left(- \frac{14}{9}\right),
\label{eq:jiajian}
\end{eqnarray} 
for $i = 1, 2, \cdots, L$, where (a) is due to monotonicity of the RIC and (b) is from \eqref{eq:RIC00}. Notice that (a) is because
\begin{eqnarray}
|{\Gamma}^k_i \cup T^{k + k_{i} - 1}| &\leq& |T \cup T^{k + k_{L}}| \nonumber \\
 &=& |T^{k + k_{L}}| + |\Gamma^{k + k_{L}}|  \nonumber \\
 &\leq& S (k + k_{L}) + |\Gamma^k|  \nonumber\\
 &\overset{(c)}{\leq}& Sk + 2 ( 2^{L} - 1) S + \gamma \label{eq:final222} \nonumber
 \\
 &\overset{(d)}{<}& Sk + 8 \left(\frac{2\sigma - 2}{2 \sigma - 1} \right) \gamma + {\gamma} - 2 S \label{eq:youxiude2} \nonumber \\
 &\overset{(e)}{<}& Sk + 7 \gamma, \label{eq:youxiude23}
\end{eqnarray}
where (c) follows from \eqref{eq:aaaaaaa2}, (d) is from \eqref{eq:next2},
and (e) is due to $\sigma = \frac{1}{2}  \exp\left(\frac{14}{9}\right)$.

For notational simplicity, we let $\eta = \exp \left(- \frac{14}{9}\right)$. Then \eqref{eq:11f2}--\eqref{eq:13f2} can be
rewritten as
\begin{eqnarray}
 && \hspace{-8mm} \| \mathbf{r}^{k + k_1} \|_2^2
  \leq \eta \| \mathbf{r}^{k} \|_2^2 + (1 - \eta)   \|\mathbf{\Phi}_{{\Gamma^{k}} \backslash {\Gamma}^k_1}  \mathbf{x}_{{\Gamma^{k}} \backslash {\Gamma}^k_1} + \mathbf{v}\|_2^2, \label{eq:11ff2}  \nonumber \\
 && \hspace{-8mm} \| \mathbf{r}^{k + k_2} \|_2^2 \leq \eta \| \mathbf{r}^{k +
k_1} \|_2^2 + (1 - \eta) \|\mathbf{\Phi}_{{\Gamma^{k}}
\backslash {\Gamma}^k_2}
  \mathbf{x}_{{\Gamma^{k}} \backslash {\Gamma}^k_2} + \mathbf{v}\|_2^2, \label{eq:12ff2}  \nonumber \\
  &&\vdots \nonumber \\
 && \hspace{-8mm} \| \mathbf{r}^{k + k_{L}} \|_2^2
  \leq \eta \| \mathbf{r}^{k + k_{{L}- 1}} \|_2^2 + (1 - \eta)  \| \mathbf{\Phi}_{{\Gamma^{k}}
\backslash {\Gamma}^k_{L}}
  \mathbf{x}_{{\Gamma^{k}} \backslash {\Gamma}^k_{L}} +
  \mathbf{v}\|_2^2.  \nonumber
\end{eqnarray}
Some additional manipulations yield the following result.
\begin{eqnarray}
 ~~&& \hspace{-9.5mm} \| \mathbf{r}^{k + k_{L}} \|_2^2 \nonumber \\ 
 && \hspace{-9.5mm} \leq \eta^{L} \| \mathbf{r}^{k} \|_2^2 + (1 - \eta) \sum_{\tau = 1}^{L}\eta^{{L}-
  \tau} \|\mathbf{\Phi}_{{\Gamma^{k}} \backslash {\Gamma}^k_\tau}
\mathbf{x}_{{\Gamma^{k}} \backslash {\Gamma}^k_\tau} + \mathbf{v}
\|_2^2 \nonumber \\
 && \hspace{-9.8mm} \overset{(a)}{\leq} \hspace{-0.5mm} \eta^{L} \|\mathbf{\Phi}_{\Gamma^k}\mathbf{x}_{\Gamma^{k}} \hspace{-0.5mm} + \hspace{-0.5mm}
\mathbf{v}\|_2^2  \hspace{-0.5mm} + \hspace{-0.5mm} (1 \hspace{-0.5mm} - \hspace{-0.5mm} \eta) \hspace{-0.5mm} \sum_{\tau = 1}^{L}  \hspace{-0.5mm} \eta^{{L}- \tau}  
\|\mathbf{\Phi}_{{\Gamma^{k}} \backslash {\Gamma}^k_\tau}
\mathbf{x}_{{\Gamma^{k}} \backslash {\Gamma}^k_\tau} \hspace{-0.5mm} + \hspace{-0.5mm}
\mathbf{v}\|_2^2  \nonumber \\
 && \hspace{-9.5mm} \overset{(b)}{\leq} \eta^{L}
((1 + t)\|\mathbf{\Phi}_{\Gamma^k}\mathbf{x}_{\Gamma^{k}}\|_2^2 + (1 + t^{-1})
\|\mathbf{v}\|_2^2) + (1 - \eta) \nonumber \\
 && \hspace{-6mm} \times \sum_{\tau = 1}^{L} \eta^{{L}-
\tau}  \left((1 + t)\|\mathbf{\Phi}_{{\Gamma^{k}} \backslash {\Gamma}^k_\tau}
\mathbf{x}_{{\Gamma^{k}} \backslash {\Gamma}^k_\tau}\|_2^2  + (1 + t^{-1}) \|\mathbf{v}\|_2^2 \right), \nonumber \\ 
&& \hspace{-9.5mm}  \overset{(c)}{\leq} \hspace{-1mm} \left( \eta^L 
\|\mathbf{x}_{\Gamma^{k} \backslash {\Gamma}^{k}_0}\|_2^2 +  (1 - \eta) \sum_{\tau = 1}^{L}\eta^{{L}- \tau} \|\mathbf{x}_{{\Gamma^{k}} \backslash {\Gamma}^k_\tau} \|_2^2 \right) \hspace{-0.5mm} (1 + t) \nonumber \\
 && \hspace{-6mm} \times (1 + \delta_{\gamma}) + (1 + t^{-1}) \left(\eta^L + (1 - \eta) \sum_{\tau = 1} ^ L \eta^{L - \tau} \right)   \|\mathbf{v}\|_2^2, \nonumber \\ \label{eq:good52373} 
\end{eqnarray}
where (a) is from Proposition \ref{prop:rk}, (b) uses the fact that 
\begin{equation}
 \|\mathbf{u} + \mathbf{v}\|_2^2 \leq (1 + t) \|\mathbf{u}\|_2^2 + (1 + t^{-1}) \|\mathbf{v}\|_2^2 \label{eq:in8}
 \end{equation} for $t>0$ (we will specify $t$ later), and (c) is due to the RIP.
(Note that $|{\Gamma^{k}} \backslash {\Gamma}^k_\tau| \leq |\Gamma^k| =
\gamma$ for $\tau = 1, \cdots, L$.)

By applying \eqref{eq:next22} to \eqref{eq:good52373},
we further have
\begin{eqnarray}
 ~~&& \hspace{-9.5mm} \| \mathbf{r}^{k + k_{L}} \|_2^2 \nonumber \\ 
 && \hspace{-9.5mm} {\leq} \hspace{-1mm}  \left( {\sigma^{L
- 1}  \eta^{L}} + (1 - \eta) \sum_{\tau = 1}^{L} \sigma ^{{L}- 1 -
\tau} \eta^{{L}- \tau}
\right) \hspace{-0.5mm} (1 + t) (1 + \delta_{\gamma}) \nonumber \\
 && \hspace{-6mm} \times  \| \mathbf{x}_{{\Gamma^{k}} \backslash {\Gamma}^k_{{L}- 1}} \|_2^2 + (1 \hspace{-0.25mm} + \hspace{-0.25mm} t^{-1}) \hspace{-0.5mm} \left(\hspace{-0.5mm} \eta^L\hspace{-0.5mm}  + (1\hspace{-0.25mm} - \hspace{-0.25mm}\eta) \sum_{\tau = 1} ^ L \eta^{L - \tau} \hspace{-0.5mm} \right) \hspace{-0.5mm}  \|\mathbf{v}\|_2^2  \nonumber \label{eq:good524} \\
 && \hspace{-9.5mm} = \hspace{-1mm}  
\left({(\sigma \eta
)^{{L}}} + (1 - \eta)  \sum_{\tau = 0}^{L - 1} (
  \sigma  \eta )^{\tau} \right) {\sigma}^{-1} (1 + \delta_{\gamma})   \| \mathbf{x}_{{\Gamma^{k}} \backslash {\Gamma}^k_{{L}- 1}}
  \|_2^2  \nonumber \\
 && \hspace{-6mm} \times (1 + t)  + (1  +   t^{-1})  \left( \eta^L + (1 - \eta) \sum_{\tau = 0} ^ {L - 1} \eta^{\tau}  \right)  \|\mathbf{v}\|_2^2   \nonumber \\
 && \hspace{-9.5mm}  \overset{(a)}{<}  \hspace{-1mm}  \left(\sum_{\tau = L}^{\infty} (
  \sigma  \eta )^{\tau} \hspace{-1mm} + \hspace{-0.75mm}  \sum_{\tau = 0}^{L - 1} (
  \sigma  \eta )^{\tau} \hspace{-1mm}  \right) \hspace{-0.5mm}  {\sigma}^{-1} (1 \hspace{-0.5mm}  - \hspace{-0.5mm}  \eta)  (1\hspace{-0.5mm}  + \hspace{-0.25mm} \delta_{\gamma}) \|  \mathbf{x}_{{\Gamma^{k}} \backslash {\Gamma}^k_{{L}- 1}} \hspace{-0.5mm}  \|_2^2 
\nonumber \\
  && \hspace{-6mm} \times  (1 + t) + (1  +   t^{-1}) (1 - \eta) \left(\sum_{\tau = L} ^ {\infty} \eta^\tau +  \sum_{\tau = 0} ^ {L - 1} \eta^{\tau}  \right)  \|\mathbf{v}\|_2^2  \label{eq:meihaojiayuan2}  \nonumber \\
   &&  \hspace{-9.5mm} \overset{(b)}{=} \hspace{-.5mm} 4 \eta(1 - \eta)(1 + \delta_{\gamma})(1 - t) \|
\mathbf{x}_{{\Gamma^{k}} \backslash {\Gamma}^k_{{L}- 1}}
  \|_2^2 + (1 + t^{-1}) \|\mathbf{v}\|_2^2, \nonumber \\ \label{eq:up}
\end{eqnarray}
where (a) is because $\sigma \geq 2$, $\sigma \eta < 1$, and $\eta < 1$. Hence
\begin{eqnarray}
  (\sigma  \eta )^{L} &<& \left(\frac{1 - \eta}{1 - \sigma \eta}\right) (\sigma
\eta)^{L} = (1 - \eta) \sum_{\tau = L}^{\infty} (\sigma
\eta)^{\tau}, \nonumber \\
  \eta^{L} &=& (1 - \eta) \left(\frac{\eta^L}{1 - \eta}\right)   = (1 - \eta) \sum_{\tau = L}^{\infty}  
\eta^{\tau}, \nonumber   
\end{eqnarray} 
and (b) uses the fact that $\sigma \eta = \frac{1}{2}$.

\vspace{1mm} 

Thus far, we have obtained a lower bound for $\|
\mathbf{r}^{k + k_{L}} \|_2$ in~\eqref{eq:low} and an upper bound for $\|
\mathbf{r}^{k + k_{L}} \|_2$ in \eqref{eq:up}, respectively. Next, we will use these bounds to prove that \eqref{eq:aasssd1} holds true
under $\delta_{Sk + 7 \gamma} \leq \frac{1}{8}$. 

By relating \eqref{eq:low} and \eqref{eq:up}, we have%
\begin{equation} \label{eq:conditiona}
    \| \mathbf{x}_{\Gamma^{k + k_{L}}} \|_2 \leq \alpha \|
  \mathbf{x}_{{\Gamma^{k}} \backslash {\Gamma}^k_{{L}- 1}} \|_2 + \beta \|\mathbf{v}\|_2
\end{equation}
where
\begin{equation}
\alpha = 2 \left({\frac{ \eta (1 - \eta)(1 + \delta_{\gamma}) (1 + t)}{1
- \delta_{|T \cup T^{k + k_{L}}|}}} \right)^{1/2}
\label{eq:347}
\end{equation} 
and
\begin{equation}
\beta = \left((1 + t^{-1})^{1/2} + 1 \right) \left(1 - \delta_{|T \cup T^{k + k_{L}}|}
\right)^{-1/2}. \label{eq:3471}
\end{equation}
Since $\delta_{|T
\cup T^{k + k_{L}}|} \leq \delta_{Sk + 7 \gamma}$ by monotonicity of the RIC,  
\begin{equation}
\alpha  \leq 2 \left(\hspace{-.5mm} {\frac{(1 + \delta_{Sk + 7 \gamma})(1 + t) \left(\hspace{-.5mm} 1 \hspace{-.5mm} -\hspace{-.5mm}  \exp \left(- \frac{14}{9} \right) \right) } {(1 - \delta_{Sk + 7 \gamma })\exp \left(\frac{14}{9}  \right)}}\hspace{-.5mm}
\right)^{1/2}. \label{eq:53o}
\end{equation}
By choosing $t = \frac{1}{6}$ in \eqref{eq:53o}, we have
\begin{equation}
\alpha < 1
\end{equation}
under $\delta_{Sk + 7 \gamma } \leq \frac{1}{8}$.

Now, we consider two cases: 1) $\beta \|\mathbf{v}\|_2 < (1 - \alpha ) \|
\mathbf{x}_{{\Gamma^{k}} \backslash {\Gamma}^k_{{L}- 1}} \|_2$ and 2) $\beta \|\mathbf{v}\|_2 \geq (1 - \alpha ) \|
\mathbf{x}_{{\Gamma^{k}} \backslash {\Gamma}^k_{{L}- 1}} \|_2$. 
First, if $\beta \|\mathbf{v}\|_2 < (1 - \alpha ) \|
\mathbf{x}_{{\Gamma^{k}} \backslash {\Gamma}^k_{{L}- 1}} \|_2$,
\eqref{eq:conditiona} implies \eqref{eq:aasssd1} (i.e., $\|
\mathbf{x}_{\Gamma^{k + k_{L}}} \|_2^2 <
\|\mathbf{x}_{\Gamma^{k} \backslash {\Gamma}^k_{{L}- 1}} \|_2^2$)
so that \eqref{eq:jjff0} holds true.

Second, if $\beta \|\mathbf{v}\|_2 \geq (1 - \alpha ) \|
\mathbf{x}_{{\Gamma^{k}} \backslash {\Gamma}^k_{{L}- 1}} \|_2$, then \eqref{eq:42} directly holds true
because
\begin{eqnarray}
  &&\hspace{-9mm} \big\|\mathbf{r}^{k + \max \left\{\gamma, \left\lfloor\frac{8 \gamma}{S} \right\rfloor \right\}} \big \|_2 \nonumber \\ 
  &&\hspace{-7mm} \overset{(a)}{\leq}  \| \mathbf{r}^{k + k_{L}} \big\|_2 \nonumber \\
  &&\hspace{-7mm}\overset{(b)}{\leq} \hspace{-.5mm} 2 \left( \eta (1 \hspace{-.5mm} - \hspace{-.5mm} \eta)  ( 1\hspace{-.5mm} + \hspace{-.5mm} \delta_{\gamma} ) (1\hspace{-.5mm} +\hspace{-.5mm} t)  \right)^{1/2} \hspace{-.5mm} \|\mathbf{x}_{{\Gamma^{k}} \backslash {\Gamma}^k_{{L}- 1}} \hspace{-.5mm} \|_2 \hspace{-.5mm}  + \hspace{-.5mm} (1 \hspace{-.5mm} + \hspace{-.5mm} t^{-1})^{1/2} \|\mathbf{v}\|_2 \nonumber \\
  &&\hspace{-7mm}\overset{(c)}{=}  \alpha (1 - \delta_{|T \cup T^{k + k_{L}}|})^{1/2} \|\mathbf{x}_{{\Gamma^{k}} \backslash
{\Gamma}^k_{{L}- 1}} \|_2 + (1 + t^{-1})^{1/2} \|\mathbf{v}\|_2 \nonumber \\
  &&\hspace{-7mm} \leq   \left( \frac{\alpha \beta (1 - \delta_{|T \cup T^{k + k_{L}}|})^{1/2}}{1 - \alpha} + (1 + t^{-1})^{1/2}\right) \|\mathbf{v}\|_2 \nonumber \\
  &&\hspace{-7mm}=  
\left(\frac{(1 + t^{-1})^{1/2} + 1}{1 - \alpha} - 1\right)\|\mathbf{v}\|_2
\nonumber \\
  &&\hspace{-7mm} {\leq}~  \mu_k \|\mathbf{v}\|_2, \label{eq:5262}
\end{eqnarray}
where 
\begin{eqnarray}   
\mu_k \hspace{-.5mm} = && \hspace{-6.65mm} \left(\hspace{-.5mm} 1 \hspace{-.5mm}- \hspace{-.5mm} 2 \left({\frac{{7} (1
\hspace{-.25mm} + \hspace{-.25mm} \delta) \left(1 \hspace{-.25mm} - \hspace{-.25mm} \exp \left(- \frac{14}{9} \right) \right) } {{6} (1 - \delta) \exp \left(\frac{14}{9} \right) }}  \right)^{\hspace{- .5mm}1/2} \right)^{\hspace{-1mm}-1} \hspace{-1mm} (\sqrt 7 \hspace{-.25mm} + \hspace{-.25mm} 1)\hspace{-.25mm}  - \hspace{-.5mm} 1 
\nonumber \\
\label{eq:ck99}
\end{eqnarray}
where $\delta = \delta_{\max \left\{Sk + 7 \gamma,  Sk + S + \gamma \right\}}$, (a) is from \eqref{eq:48} and the fact that the residual power of gOMP is always non-increasing, (b) is due to \eqref{eq:up}, and (c) is from~\eqref{eq:347}.

\subsection{Proof of \eqref{eq:jjff00}}
The proof of \eqref{eq:jjff00} is similar to the proof of \eqref{eq:jjff0}. Instead of directly proving \eqref{eq:jjff00}, we will show that a sufficient condition for
\eqref{eq:jjff00} is true. More precisely, we will prove that 
\begin{equation} \label{eq:aasssd10}
\| \mathbf{x}_{\Gamma^{k + 1}} \|_2^2 <
\|\mathbf{x}_{\Gamma^{k}} \|_2^2
\end{equation} 
holds true under \begin{equation}
\delta_{S(k + 2) + \gamma} \leq \frac{1}{8}.
\end{equation}
We first construct lower and upper
bounds for $\| \mathbf{r}^{k + 1} \|_2$ and then use these
bounds to derive a condition guaranteeing \eqref{eq:aasssd10}. 

\vspace{1mm}
\noindent \textbf{1) Lower bound for $\| \mathbf{r}^{k + 1}\|_2^2$}:
\begin{eqnarray}
\lefteqn{\| \mathbf{r}^{k + 1} \|_2} \nonumber \\ 
\hspace{-2mm} &=& \hspace{-2mm} \|\mathbf{y} - \mathbf{\Phi}  \hat{\mathbf{x}}^{k + 1} \|_2 \nonumber \\
 \hspace{-2mm} &=& \hspace{-2mm} \| \mathbf{\Phi} ( \mathbf{x} - \hat{\mathbf{x}}^{k + 1} )  + \mathbf{v} \|_2 \nonumber  \\
 \hspace{-2mm} &\geq& \hspace{-2mm} \| \mathbf{\Phi} (\mathbf{x} - \hat{\mathbf{x}}^{k + 1}) \|_2 - \| \mathbf{v} \|_2 \nonumber  \\
 \hspace{-2mm} &\overset{(a)}{\geq}& \hspace{-2mm} \left(1 - \delta_{|T \cup T^{k + 1}|}\right)^{1/2} \| \mathbf{x} - \hat{\mathbf{x}}^{k + 1}  \|_2 - \| \mathbf{v} \|_2  \nonumber   \\
 \hspace{-2mm} &\geq&  \hspace{-2mm}\left(1 - \delta_{|T \cup T^{k + 1}|}\right)^{1/2} \| \mathbf{x}_{\Gamma^{k + 1}} \|_2 - \| \mathbf{v}\|_2,  \label{eq:low0}
\end{eqnarray}
where (a) is because $\mathbf{x} - \hat{\mathbf{x}}^{k + 1}$ is supported on $T \cup T^{k + 1}$.

\vspace{1mm}
\noindent \textbf{2) Upper bound for $\|\mathbf{r}^{k + 1}\|_2^2$}:

\vspace{1mm}
By applying Proposition \ref{prop:rk} with $l = k$ and $\tau = 1$, we have
\begin{eqnarray}  
\lefteqn{ \hspace{-3mm}  \|\mathbf{r}^k\|_2^2 - \|\mathbf{r}^{k + 1}\|_2^2} \nonumber \\
&& \hspace{-6mm} \geq 
\frac{1 - \delta_{|{\Gamma}^k_{1} \cup T^k |}} {(1 + \delta_{S})
\left\lceil \frac{|{\Gamma}^k_1 |}{S} \right\rceil}  \left( \|\mathbf{r}^k \|_2^2 -
\|\mathbf{\Phi}_{{\Gamma^{k}} \backslash {\Gamma}^k_1}
\mathbf{x}_{{\Gamma^{k}} \backslash {\Gamma}^k_1} + \mathbf{v}
\|_2^2 \right) \nonumber \\
&& \hspace{-6mm} \overset{(a)}{=} 
\frac{1 - \delta_{|{\Gamma}^k_{1} \cup T^k |}} {1 + \delta_{S}} \left( \|\mathbf{r}^k \|_2^2 -
\|\mathbf{\Phi}_{{\Gamma^{k}} \backslash {\Gamma}^k_1}
\mathbf{x}_{{\Gamma^{k}} \backslash {\Gamma}^k_1} + \mathbf{v}
\|_2^2 \right), \label{eq:residual3goodaaa}
\end{eqnarray}
where (a) uses the fact that $|\Gamma^k_1| \leq S$ (see \eqref{eq:jjjjffff}) and hence $\left\lceil \frac{|{\Gamma}^k_1 |}{S} \right\rceil = 1$. Rearranging the terms yields
\begin{eqnarray}  
\lefteqn{\|\mathbf{r}^{k + 1}\|_2^2 \leq \left(1 -  \frac{1 - \delta_{|{\Gamma}^k_{1} \cup T^k |}} {1 + \delta_{S}}  \right)   \|\mathbf{r}^k \|_2^2 } \nonumber \\
&& \hspace{6mm} + \frac{1 - \delta_{|{\Gamma}^k_{1} \cup T^k |}} {1 + \delta_{S}} \|\mathbf{\Phi}_{{\Gamma^{k}} \backslash {\Gamma}^k_1}
\mathbf{x}_{{\Gamma^{k}} \backslash {\Gamma}^k_1} + \mathbf{v}
\|_2^2. \label{eq:62eq}
\end{eqnarray}
From Proposition \ref{prop:rk}, 
\begin{eqnarray}  
\|\mathbf{r}^k \|_2^2 &\leq& \|\mathbf{\Phi}_{\Gamma^k} \mathbf{x}_{\Gamma^k} + \mathbf{v}\|_2^2 \nonumber \\
&\overset{(a)}{\leq}& (1 + t) \|\mathbf{\Phi}_{\Gamma^k} \mathbf{x}_{\Gamma^k} \|_2^2 + (1 + t^{-1}) \|\mathbf{v}\|_2^2 \nonumber \\
&\overset{(b)}{\leq}& (1 + t) (1 + \delta_\gamma) \| \mathbf{x}_{\Gamma^k} \|_2^2 + (1 + t^{-1}) \|\mathbf{v}\|_2^2, ~~~\label{eq:63eq}
\end{eqnarray}
where (a) is from \eqref{eq:in8} and (b) is due to the RIP. 
Moreover,
\begin{eqnarray}
\lefteqn{\|\mathbf{\Phi}_{{\Gamma^{k}} \backslash {\Gamma}^k_1}
\mathbf{x}_{{\Gamma^{k}} \backslash {\Gamma}^k_1} + \mathbf{v}
\|_2^2} \nonumber \\
&\overset{(a)}{\leq}& (1 + t) \|\mathbf{\Phi}_{{\Gamma^{k}} \backslash {\Gamma}^k_1}
\mathbf{x}_{{\Gamma^{k}} \backslash {\Gamma}^k_1}\|_2^2 + (1 + t^{-1}) \| \mathbf{v}
\|_2^2 \nonumber \\
&\overset{(b)}{\leq}& (1 + t) (1 + \delta_\gamma) \|\mathbf{x}_{{\Gamma^{k}} \backslash {\Gamma}^k_1}\|_2^2 + (1 + t^{-1}) \| \mathbf{v}
\|_2^2 \nonumber \\
&\overset{(b)}{\leq}& (1 + t) (1 + \delta_\gamma) \sigma^{-1} \|\mathbf{x}_{{\Gamma^{k}}}\|_2^2 + (1 + t^{-1}) \| \mathbf{v}
\|_2^2,~~ \label{eq:64eq}
\end{eqnarray}
where (a) is from \eqref{eq:in8}, (b) is due to the RIP, and (c) is from \eqref{eq:mu4}.

Using \eqref{eq:62eq}, \eqref{eq:63eq}, and \eqref{eq:64eq}, we have
\begin{eqnarray}
&& \hspace{-7mm} \|\mathbf{r}^{k + 1}\|_2^2 \leq (1 + t) (1 + \delta_\gamma) \nonumber \\
&& \hspace{-7mm} \times \left(1 -  \frac{(1 - \sigma^{-1})  (1 - \delta_{|{\Gamma}^k_{1} \cup T^k |})} {1 + \delta_{S}}  \right) \|\mathbf{x}_{{\Gamma^{k}}}\|_2^2 + (1 + t^{-1}) \|\mathbf{v}\|_2^2, \nonumber  
\end{eqnarray}
from which we obtain an upper bound for $\|\mathbf{r}^{k + 1}\|_2$ as 
\begin{eqnarray}
&& \hspace{-5mm} \|\mathbf{r}^{k + 1}\|_2 \leq (1 + t)^{1/2} \left(1 -  \frac{(1 - \sigma^{-1})  (1 - \delta_{|{\Gamma}^k_{1} \cup T^k |})} {1 + \delta_{S}}  \right)^{1/2} \nonumber \\
&& \hspace{12mm} \times (1 + \delta_\gamma)^{1/2}  \|\mathbf{x}_{{\Gamma^{k}}}\|_2 + (1 + t^{-1})^{1/2} \|\mathbf{v}\|_2.  \label{eq:upper0}
\end{eqnarray}

\vspace{1mm}
Thus far, we have established a lower bound for $\|
\mathbf{r}^{k + 1} \|_2$ in \eqref{eq:low0} and an upper bound for $\|
\mathbf{r}^{k + 1} \|_2$ in \eqref{eq:upper0}. Now we combine \eqref{eq:low0} and \eqref{eq:upper0} to obtain
\begin{equation}
\| \mathbf{x}_{\Gamma^{k + 1}} \|_2 \leq \alpha' \|\mathbf{x}_{{\Gamma^{k}}}\|_2 + \beta' \|\mathbf{v}\|_2, \label{eq:65io}
\end{equation}
where
\begin{eqnarray}
&& \hspace{-10mm}  \alpha' = \left( \frac{(1\hspace{-.5mm} + t) (1 \hspace{-.5mm} + \delta_\gamma)}{1 - \delta_{|T \cup T^{k + 1}|}} \hspace{-.5mm} \left(\hspace{-.25mm}1 \hspace{-.5mm}- \hspace{-.5mm} \frac{(1\hspace{-.5mm} - \sigma^{-1})  (1\hspace{-.5mm} - \delta_{|{\Gamma}^k_{1} \cup T^k |})} {1 + \delta_{S}}  \right) \right)^{\hspace{-.5mm}1/2}  \label{eq:alpha'}
\end{eqnarray}
and 
\begin{equation}
\beta' = \left((1 + t^{-1})^{1/2} + 1 \right) \left(1 - \delta_{|T \cup T^{k + 1}|}
\right)^{-1/2}. 
\end{equation}
 Recalling that $t = \frac{1}{6}$ and $\sigma = \frac{1}{2}  \exp\left(\frac{14}{9}\right)$ and also noting that $\delta_{|{\Gamma}^k_{1} \cup T^k |} \leq \delta_{|T \cup T^{k + 1}|} = \delta_{|T^{k + 1}| + |\Gamma^{k + 1}|} \leq  \delta_{Sk + S + \gamma}$ and $\delta_{\gamma} \leq  \delta_{Sk + S + \gamma}$, one can show from \eqref{eq:alpha'} that \begin{equation}
 \alpha' < 1
 \end{equation}
under $\delta_{Sk + S + \gamma} \leq \frac{1}{8}$.
\vspace{1mm}

Now, we consider two cases: 1) $\beta' \|\mathbf{v}\|_2 < (1 - \alpha' ) \| \mathbf{x}_{{\Gamma^{k}}}\|_2$ and 2) $\beta' \|\mathbf{v}\|_2 \geq (1 - \alpha' ) \| \mathbf{x}_{{\Gamma^{k}}}\|_2$. 
First, if $\beta' \|\mathbf{v}\|_2 < (1 - \alpha' ) \| \mathbf{x}_{{\Gamma^{k}}}\|_2$, \eqref{eq:65io} implies \eqref{eq:aasssd10} (i.e., $\| \mathbf{x}_{\Gamma^{k + 1}} \|_2^2 <
\|\mathbf{x}_{\Gamma^{k}} \|_2^2$)
so that \eqref{eq:jjff00} holds true.

Second, if
$\beta' \|\mathbf{v}\|_2 \geq (1 - \alpha' ) \| \mathbf{x}_{{\Gamma^{k}}}\|_2$, then \eqref{eq:42} directly holds true
because
\begin{eqnarray} 
  &&\hspace{-10mm} \big\|\mathbf{r}^{k + \max \left\{\gamma, \left\lfloor\frac{8 \gamma}{S} \right\rfloor \right\}} \big \|_2 \nonumber \\ 
  &&\hspace{-7mm} \overset{(a)}{\leq}  \| \mathbf{r}^{k + 1} \big\|_2 \nonumber \\ 
  &&\hspace{-7mm}\overset{(c)}{=}  \alpha' (1 - \delta_{Sk + S + \gamma})^{1/2} \|\mathbf{x}_{{\Gamma^{k}} } \|_2 + (1 + t^{-1})^{1/2} \|\mathbf{v}\|_2 \nonumber \\
  &&\hspace{-7mm} \leq   \left( \frac{\alpha' \beta' (1 - \delta_{Sk + S + \gamma})^{1/2}}{1 - \alpha'} + (1 + t^{-1})^{1/2}\right) \|\mathbf{v}\|_2 \nonumber \\
  &&\hspace{-7mm} =  
\left(\frac{(1 + t^{-1})^{1/2} + 1}{1 - \alpha'} - 1\right)\|\mathbf{v}\|_2
\nonumber \\
  &&\hspace{-7mm} {\leq} ~ \mu_k \|\mathbf{v}\|_2,
\end{eqnarray}
where (a) is due to \eqref{eq:90} and the fact that the residual power of gOMP is always non-increasing and (b) is from~\eqref{eq:upper0} and~\eqref{eq:alpha'}.
This completes the proof of \eqref{eq:jjff00}.


%
%
%
%
%
%
%


\section{Conclusion} \label{sec:conclusion}

As a method to enhance the recovery performance of orthogonal matching pursuit (OMP), generalized (gOMP) has
received attention in recent years~\cite{liu2012orthogonal,wang2012Generalized,liu2012super,huang2011recovery,maleh2011improved,satpathi2013improving,shen2014analysis,dan2014analysis,li2015sufficient}. While empirical evidence has shown that gOMP is effective in reconstructing sparse signals, theoretical results to date are relatively weak. In this paper, we have presented improved recovery guarantee of gOMP by showing that the gOMP algorithm can perform stable recovery of all
sparse signals from the noisy measurements under the restricted isometry property (RIP) with $\delta_{\max \left\{9, S + 1 \right\}K} \leq \frac{1}{8}$.
The presented proof strategy might be
useful for obtaining improved results for other greedy algorithms
derived from the OMP algorithm.

\appendices

%
%
%
%
%

\section{Proof of Theorem \ref{cor:cor2}} \label{app:cor2}

\begin{proof} We first give the proof of \eqref{eq:10ooo}. Observe that
\begin{eqnarray}
 && \hspace{-7mm} \left\|\mathbf{r}^{\max \left\{K, \left\lfloor\frac{8 K}{S} \right\rfloor \right\}} \right\|_2 \nonumber \\ 
 && \hspace{-7mm} = \left\|\mathbf{y} - \mathbf{\Phi}
\hat{\mathbf{x}}^{\max \left\{K, \left\lfloor\frac{8 K}{S} \right\rfloor \right\}} \right\|_2 \nonumber \\
 && \hspace{-7mm} =  \left\|\mathbf{\Phi} \left(\mathbf{x} - \hat{\mathbf{x}}^{\max \left\{K, \left\lfloor\frac{8 K}{S} \right\rfloor \right\}} \right) + \mathbf{v} \right\|_2 \nonumber
 \\
 && \hspace{-7mm} \geq \left\|\mathbf{\Phi}\left(\mathbf{x} - \hat{\mathbf{x}}^{ \max \left\{K, \left\lfloor\frac{8 K}{S} \right\rfloor \right\}} \right) \right\|_2 -
 \|\mathbf{v}\|_2 \nonumber \\
 && \hspace{-7mm} \overset{(a)}{\geq} \hspace{-.5mm} \left(\hspace{-.5mm}1 \hspace{-.5mm} - \hspace{-.5mm} \delta_{\big|T \cup T^{\max \left\{K, \left\lfloor\frac{8 K}{S} \right\rfloor \right\}} \big|}   \right)^{\hspace{-.5mm}1/2} \hspace{-.5mm} \left\|\mathbf{x} \hspace{-.5mm} - \hspace{-.5mm} \hat{\mathbf{x}}^{\max \left\{K, \left\lfloor\frac{8 K}{S} \right\rfloor \right\}} \hspace{-.5mm} \right\|_2 \hspace{-.5mm} - \hspace{-.5mm}
 \|\mathbf{v}\|_2  \nonumber \\
 && \hspace{-7mm} \overset{(b)}{\geq} \hspace{-.5mm} \left (1 \hspace{-.5mm} -  \hspace{-.5mm}\delta_{\max \left\{9, S + 1 \right\}K} \right)^{1/2} \hspace{-.5mm} \left\|\mathbf{x} \hspace{-.5mm} - \hspace{-.5mm} \hat{\mathbf{x}}^{\max \left\{K, \left\lfloor\frac{8 K}{S} \right\rfloor \right\}} \hspace{-.5mm} \right\|_2 \hspace{-.5mm} - \hspace{-.5mm}
 \|\mathbf{v}\|_2. \nonumber \\
 \label{eq:tf}
\end{eqnarray}
where (a) is from the RIP and (b) is because
\begin{eqnarray}
\big|T \cup T^{\max \left\{K, \left\lfloor\frac{8 K}{S} \right\rfloor \right\}}\big| 
& \leq & |T| + \big|T^{\max \left\{K, \left\lfloor\frac{8 K}{S} \right\rfloor \right\}}\big|  \nonumber \\
& {\leq} & K + \max \left\{K, \left\lfloor\frac{8 K}{S} \right\rfloor \right\} S \nonumber \\ 
& \leq & \max \{9, S + 1 \}K. \nonumber
\end{eqnarray}
Using \eqref{eq:42} and \eqref{eq:tf}, we have
\begin{eqnarray}
 && \hspace{-6.5mm} \big\|\mathbf{x} - \hat{\mathbf{x}}^{\max \left\{K, \left\lfloor\frac{8 K}{S} \right\rfloor \right\}} \big\|_2 \nonumber \\
 && \hspace{-4.5mm}\leq \left (1 - \delta_{\max \left\{9, S + 1 \right\}K} \right)^{-1/2}  \left( \big\|\mathbf{r}^{\max \left\{K, \left\lfloor\frac{8 K}{S} \right\rfloor \right\}} \big\|_2 +  \|\mathbf{v}\|_2 \right) \nonumber \\ 
 && \hspace{-5mm}\overset{(a)}{\leq} \left (1 - \delta_{\max \left\{9, S + 1 \right\}K} \right)^{-1/2} (\mu_0 + 1)\|\mathbf{v}\|_2
 \nonumber \\
 && \hspace{-4.5mm} = ~\mu \|\mathbf{v}\|_2, \label{eq:61}
\end{eqnarray} 
where 
\begin{eqnarray}   
\mu = && \hspace{-6mm} \left(\hspace{-.5mm} 1 \hspace{-.5mm}- \hspace{-.5mm} 2 \left({\frac{{7} (1
\hspace{-.25mm} + \hspace{-.25mm} \delta) \left(1 \hspace{-.25mm} - \hspace{-.25mm} \exp \left(- \frac{14}{9} \right) \right) } {{6} (1 - \delta) \exp \left(\frac{14}{9} \right) }}  \right)^{\hspace{- .5mm}1/2} \right)^{\hspace{-1mm}-1} \hspace{-1mm} \frac{\sqrt 7 \hspace{-.25mm} + \hspace{-.25mm} 1}{(1 - \delta)^{-1/2}} \hspace{-.25mm}  \nonumber \\ 
\label{eq:ck991}
\end{eqnarray}
where $\delta= \delta_{\max \left\{9, S + 1 \right\}K}$ and (a) is because $\delta_{7K} 
\leq \delta_{\max \left\{9, S + 1 \right\}K} \leq \frac{1}{8}$ (see Theorem~\ref{cor:cor2}).
\vspace{1mm}

Now, we turn to the proof of \eqref{eq:10ooo1}. 
Using the best $K$-term approximation $\left (\hat{\mathbf{x}}^{\max \left\{K, \left\lfloor\frac{8 K}{S} \right\rfloor \right\}} \right )_K$ of $\hat{\mathbf{x}}^{\max \left\{K, \left\lfloor\frac{8 K}{S} \right\rfloor \right\}}$,
we have
\begin{eqnarray}
\hspace{-5mm}&&\hspace{-7mm} \left \| \left( \hat{\mathbf{x}}^{\max \left\{K, \left\lfloor\frac{8 K}{S} \right\rfloor \right\}} \right)_K - \mathbf{x} \right\|_2 
\nonumber \\
&&\hspace{-7mm} = \left \| \hspace{-.5mm} \left ( \hat{\mathbf{x}}^{\max \left\{K, \left\lfloor\frac{8 K}{S} \right\rfloor \right\}} \hspace{-.5mm} \right )_K \hspace{-2mm} - \hspace{-.5mm} \hat{\mathbf{x}}^{ \max \left\{K, \left\lfloor\frac{8 K}{S} \right\rfloor \right\}} \hspace{-1mm} + \hspace{-.5mm} \hat{\mathbf{x}}^{ \max \left\{K, \left\lfloor\frac{8 K}{S} \right\rfloor \right\}} \hspace{-.5mm} - \hspace{-.5mm} \mathbf{x}  \right\|_2  
\nonumber \\
&&\hspace{-7mm} \overset{(a)}{\leq}  \left \| \left(\hat{\mathbf{x}}^{ \max \left\{K, \left\lfloor\frac{8 K}{S} \right\rfloor \right\}} \right)_K - \hat{\mathbf{x}}^{ \max \left\{K, \left\lfloor\frac{8 K}{S} \right\rfloor \right\}} \right\|_2  \nonumber \\
&& \hspace{-2mm} + \left \|\hat{\mathbf{x}}^{ \max \left\{K, \left\lfloor\frac{8 K}{S} \right\rfloor \right\}} - \mathbf{x} \right\|_2   \nonumber \\
&&\hspace{-7mm} \overset{(b)}{\leq}  2 \left \|\hat{\mathbf{x}}^{ \max \left\{K, \left\lfloor\frac{8 K}{S} \right\rfloor \right\}} - \mathbf{x} \right\|_2 \nonumber \\
&&\hspace{-6.5mm} \leq 2 \mu \|\mathbf{v} \|_2,  \label{eq:zuuuo}
\end{eqnarray} 
where (a) is from the triangle inequality and (b) is because $\left(\hat{\mathbf{x}}^{ \max \left\{K, \left\lfloor\frac{8 K}{S} \right\rfloor \right\}} \right)_K$ is the best $K$-term approximation to $\hat{\mathbf{x}}^{ \max \left\{K, \left\lfloor\frac{8 K}{S} \right\rfloor \right\}}$ and hence is a better approximation than $\mathbf{x}$ (note that both $\left(\hat{\mathbf{x}}^{ \max \left\{K, \left\lfloor\frac{8 K}{S} \right\rfloor \right\}} \right)_K$ and $\mathbf{x}$ are  $K$-sparse).

On the other hand, 
\begin{eqnarray}
&& \hspace{-7mm} \left\| \left(\hat{\mathbf{x}}^{ \max \left\{K, \left\lfloor\frac{8 K}{S} \right\rfloor \right\}} \right)_K - \mathbf{x} \right\|_2 \nonumber \\
&& \hspace{-7mm} \overset{(a)}{\geq} (1 - \delta_{2K})^{-1/2}  \left \|\mathbf{\Phi}  \left (  \left (\hat{\mathbf{x}}^{ \max \left\{K, \left\lfloor\frac{8 K}{S} \right\rfloor \right\}}  \right)_K - \mathbf{x} \right) \right \|_2  \nonumber \\
&& \hspace{-6.6mm}= (1 - \delta_{2K})^{-1/2}  \left\|\mathbf{\Phi}  \left(\hat{\mathbf{x}}^{ \max \left\{K, \left\lfloor\frac{8 K}{S} \right\rfloor \right\}}  \right)_K - \mathbf{y} + \mathbf{v} \right\|_2  \nonumber \\
&& \hspace{-7mm}\overset{(b)}{\geq}  (1 - \delta_{2K})^{-1/2} \left( \left\|\mathbf{\Phi} \left(\hat{\mathbf{x}}^{ \max \left\{K, \left\lfloor\frac{8 K}{S} \right\rfloor \right\}} \right)_K - \mathbf{y} \right\|_2 - \|\mathbf{v}\|_2  \right) \nonumber \\
&& \hspace{-7mm}\overset{(c)}{\geq}  (1 - \delta_{2K})^{-1/2} \left( \|\mathbf{\Phi} \hat{\mathbf{x}} - \mathbf{y}\|_2 - \|\mathbf{v}\|_2  \right) \nonumber \\
&& \hspace{-6.6mm} =  (1 - \delta_{2K})^{-1/2} \left( \|\mathbf{\Phi} (\hat{\mathbf{x}} - \mathbf{x}) - \mathbf{v}\|_2 - \|\mathbf{v}\|_2  \right) \nonumber \\
&& \hspace{-7mm}\overset{(d)}{\geq}  (1 - \delta_{2K})^{-1/2} \left( \|\mathbf{\Phi} (\hat{\mathbf{x}} - \mathbf{x})\|_2  - 2 \|\mathbf{v}\|_2  \right) \nonumber \\
&& \hspace{-7mm}\overset{(e)}{\geq}  (1 - \delta_{2K})^{-1/2} \left( (1 + \delta_{2K})^{1/2} \|\hat{\mathbf{x}} - \mathbf{x}\|_2 - 2 \|\mathbf{v}\|_2  \right) \nonumber \\
&& \hspace{-6.5mm} {\geq}~  (1 - \delta_{\max \left\{9, S + 1 \right\}K})^{-1/2} \nonumber \\
&& \hspace{-3mm} \times \left( (1 + \delta_{\max \left\{9, S + 1 \right\}K})^{1/2} \|\hat{\mathbf{x}} - \mathbf{x}\|_2 - 2 \|\mathbf{v}\|_2  \right),   \label{eq:79eq}
\end{eqnarray}
where (a) is from the RIP, (b) and (d) are from the triangle inequality, (c) is because $\left(\hat{\mathbf{x}}^{ \max \left\{K, \left\lfloor\frac{8 K}{S} \right\rfloor \right\}} \right)_K$ is supported on $\hat{T}$ and $$\hat{\mathbf{x}}_{\hat{{T}}} = \mathbf{\Phi}^\dag_{\hat{{T}}} \mathbf{y} = \arg
\underset{\mathbf{u}}{ \min} \|\mathbf{y} - \mathbf{\Phi}_{\hat{{T}}} \mathbf{u}\|_2,$$ and (e) follows from the RIP. 

Combining \eqref{eq:zuuuo} and \eqref{eq:79eq} yields  
\begin{eqnarray}
&& \hspace{-9mm} (1 - \delta_{\max \left\{9, S + 1 \right\}K})^{-1/2} \nonumber \\
&& \hspace{-9mm} \times \left( (1 + \delta_{\max \left\{9, S + 1 \right\}K})^{1/2} \|\hat{\mathbf{x}} - \mathbf{x}\|_2 - 2 \|\mathbf{v}\|_2  \right) \leq 2 \mu \|\mathbf{v} \|_2. \nonumber
\end{eqnarray}
That is, 
\begin{equation}
\|\hat{\mathbf{x}} - \mathbf{x}\|_2 \leq C \|\mathbf{v}\|_2,
\end{equation}
where \begin{eqnarray} \label{eq:Cvalue}
C \hspace{-2.5mm}&=&\hspace{-2.5mm} 2\left( \frac{1 \hspace{-.5mm} + \hspace{-.5mm} \delta_{\max \left\{9, S + 1 \right\}K}}{1 \hspace{-.5mm} - \hspace{-.5mm} \delta_{\max \left\{9, S + 1 \right\}K}} \right)^{\hspace{-.5mm}1/2} \hspace{-2.5mm} \mu \hspace{-.5mm} + \hspace{-.5mm} 2\left({1\hspace{-.5mm} - \hspace{-.5mm} \delta_{\max \left\{9, S + 1 \right\}K}} \right)^{-1/2} \nonumber \\
\hspace{-2.5mm}&=&\hspace{-2.5mm}  \frac{2 (1 + \delta)^{1/2}}{1 - \delta}  \left(\hspace{-.5mm} 1 \hspace{-.5mm}- \hspace{-.5mm} 2 \left({\frac{{7} (1
\hspace{-.25mm} + \hspace{-.25mm} \delta) \left(1 \hspace{-.25mm} - \hspace{-.25mm} \exp \left(- \frac{14}{9} \right) \right) } {{6} (1 - \delta) \exp \left(\frac{14}{9} \right) }}  \right)^{\hspace{- .5mm}1/2} \right)^{\hspace{-1mm}-1}  \nonumber \\
&&\hspace{-2mm} \times\hspace{-.25mm} (\sqrt 7\hspace{-.25mm} +\hspace{-.25mm} 1)\hspace{-.25mm} + \hspace{-.25mm} 2\left({1 \hspace{-.25mm} - \hspace{-.25mm} \delta} \right)^{-1/2}
\end{eqnarray}
where $\delta= \delta_{\max \left\{9, S + 1 \right\}K}$,
which completes the proof.

\end{proof}

%
%
%
%

\section{Proof of \eqref{eq:next22}} \label{app:relations}
\begin{proof}
Recall from \eqref{eq:mu1} that \begin{equation}
\|\mathbf{x}_{{\Gamma^{k}}
\backslash {\Gamma}^k_{L-2}}\|_2^2 < \sigma
 \|\mathbf{x}_{{\Gamma^{k}} \backslash {\Gamma}^k_{L-1}}\|_2^2.
\end{equation} 
Subtracting both sides by $\|\mathbf{x}_{{\Gamma^{k}} \backslash {\Gamma}^k_{L-1}}\|_2^2$, we have
\begin{equation} \label{eq:mu2}
\|\mathbf{x}_{{\Gamma}^k_{{L}- 1} \backslash {\Gamma}^k_{{L}-
2}}\|_2^2 < (\sigma  - 1) \|\mathbf{x}_{{\Gamma^{k}} \backslash
{\Gamma}^k_{{L}- 1}}\|_2^2.
\end{equation}
Since $|x_1| \geq |x_2| \geq \cdots \geq |x_{|\Gamma^{k}|}|$ and also noting that $\Gamma^{k} \backslash {\Gamma}^k_{{L}- 1} = \left\{2^{L - 2} + 1, \cdots, |\Gamma^{k}|\right\}$ (see \eqref{eq:jjjjffff}), the elements of $\mathbf{x}_{\Gamma^{k} \backslash {\Gamma}^k_{{L}- 1}}$
are $|\Gamma^{k}| - 2^{{L}- 2} S$ smallest ones (in magnitude) of the vector $\mathbf{x}_{\Gamma^{k}}$.
Furthermore, since $\sigma  - 1 \geq 1$, \eqref{eq:mu2} is equivalent to 
\begin{equation} \label{eq:lager}
  |{\Gamma}^k_{{L}- 1} \backslash {\Gamma}^k_{{L}- 2}| < (\sigma  - 1) (|\Gamma^{k}| - 2^{{L}- 2} S).
\end{equation}
Now we consider two cases. First, when $L = 2$, one can rewrite \eqref{eq:lager} as
\begin{equation}
|\Gamma^k_{L - 1}| < (\sigma - 1) (|\Gamma^k| - S),
\end{equation}
and hence 
\begin{equation}
|\Gamma^k| > \left(\frac{\sigma}{\sigma - 1}\right)S. 
\end{equation}
Second, when $L \geq 3$, \eqref{eq:lager} becomes
\begin{equation}
  2^{{L}- 3} S < (\sigma  - 1) ({|\Gamma^{k}|} - 2^{{L}- 2} S).
\end{equation}
Equivalently, 
\begin{equation}
|\Gamma^k| > \left(\frac{2\sigma - 1}{2\sigma - 2}\right) 2^{L - 2} S. 
\end{equation}
Combining these two cases yields the desired result.   
\end{proof}

\section{Proof of Proposition \ref{prop:rk}} \label{app:prop1}
\begin{proof}
We first consider the proof of \eqref{eq:residual3}. 
$(T^{k} \cap T) \subseteq T^{k}$ implies that
\begin{equation} \label{eq:126eq}
\|\mathbf{r}^{k}\|_2^2 = \|\mathcal{P}_{T^{k}}^{\bot} \mathbf{y}
\|_2^2 \leq \|\mathcal{P}_{T^{k} \cap T}^{\bot} \mathbf{y} \|_2^2.
\end{equation}
Also, noting that $\mathcal{P}_{T^{k} \cap T}^{\bot} \mathbf{y}$ is
the projection of $\mathbf{y}$ onto the orthogonal complement of
$span(\mathbf{\Phi}_{T^{k} \cap T})$,
\begin{equation} \label{eq:127eq}
\|\mathcal{P}_{T^{k} \cap T}^{\bot} \mathbf{y} \|_2^2 =
\min_{supp(\mathbf{z}) = T^{k} \cap T} {\|\mathbf{y}-\mathbf{\Phi}
\mathbf{z}\|}_2^2.
\end{equation}
From \eqref{eq:126eq} and \eqref{eq:127eq}, we have
\begin{eqnarray}
\| \mathbf{r}^{k} \|_2^2 &\leq&{\|\mathbf{y}-\mathbf{\Phi}_{T^{k}
\cap T}
\mathbf{x}_{T^{k} \cap T} \|}_2^2  \nonumber \\
&=& {\|\mathbf{\Phi}_T \mathbf{x}_T + \mathbf{v} - \mathbf{\Phi}_
{T^{k} \cap T}
\mathbf{x}_{T^{k} \cap T} \|}_2^2 \label{eq:ggsssss11533}  \nonumber \\
&=& {\|\mathbf{\Phi}_{\Gamma^{k}} \mathbf{x}_{\Gamma^{k}} +
\mathbf{v}\|}_2^2, \label{eq:ggsssss115}
\end{eqnarray}
where \eqref{eq:ggsssss115} is from $T \backslash (T^{k} \cap T) = T
\backslash T^{k} = \Gamma^{k}$.

Now, we turn to the proof of \eqref{eq:residual31}. The proof consists of two steps. First, we will show that the residual power difference of the gOMP satisfies
\begin{equation} \label{eq:residual_A1}
   \|\mathbf{r}^l\|_2^2 - \|\mathbf{r}^{l + 1}\|_2^2 \geq \frac{1}{1 + \delta_{S}} \|\mathbf{\Phi}'_{\Lambda^{k + 1}}
   \mathbf{r}^l\|_2^2.
\end{equation}
Second, we will show that
\begin{equation} \label{eq:nalia}
   \|\mathbf{\Phi}'_{\Lambda^{l + 1}} \mathbf{r}^l\|_2^2 \geq \frac{1 - \delta_{|{\Gamma}^k_{\tau} \cup T^l |}} {\left\lceil \frac{|{\Gamma}^k_\tau |}{S} \right\rceil}\left( \|\mathbf{r}^l \|_2^2 - \|\mathbf{\Phi}_{{\Gamma^{k}} \backslash {\Gamma}^k_\tau} \mathbf{x}_{{\Gamma^{k}} \backslash {\Gamma}^k_\tau} + \mathbf{v}\|_2^2
   \right).
\end{equation}
\eqref{eq:residual31} is established by combining \eqref{eq:residual_A1} and
\eqref{eq:nalia}.

\vspace{2mm}
\begin{itemize}
  \item {\bf Proof of \eqref{eq:residual_A1}}:

  Recall that the gOMP algorithm orthogonalizes the measurements
$\mathbf{y}$ against previously chosen columns of $\mathbf{\Phi}$,
yielding the updated residual in each iteration. That is,
\begin{equation}
  \mathbf{r}^{l + 1} = \mathcal{P}^\bot_{T^{l + 1}} \mathbf{y}.
\end{equation}
Since $ \mathbf{r}^l = \mathbf{y} - \mathbf{\Phi}
\hat{\mathbf{x}}^l$, we have
\begin{eqnarray}
  \mathbf{r}^{l + 1} = \mathcal{P}^\bot_{T^{k + 1}} (\mathbf{r}^l + \mathbf{\Phi} \hat{\mathbf{x}}^l) = \mathcal{P}^\bot_{T^{l + 1}} \mathbf{r}^l. \label{eq:ok}
\end{eqnarray}
where \eqref{eq:ok} is because $\mathbf{\Phi} \hat{\mathbf{x}}^l \in
span(\mathbf{\Phi}_{T^{l}})$ and $T^l \subset T^{l + 1}$ and hence
$\mathcal{P}^\bot_{T^{l + 1}} \mathbf{\Phi} \hat{\mathbf{x}}^l =
\mathbf{0}$. As a result,
\begin{equation}
  \mathbf{r}^l - \mathbf{r}^{l + 1} = \mathbf{r}^l - \mathcal{P}^\bot_{T^{l + 1}}
  \mathbf{r}^l = \mathcal{P}_{T^{l + 1}} \mathbf{r}^l.
\end{equation}
Noting that $\Lambda^{l + 1} \subseteq T^{l + 1}$, we have
\begin{eqnarray} \label{eq:residual_A}
   \|\mathbf{r}^l - \mathbf{r}^{l + 1}\|_2 = \|\mathcal{P}_{T^{l + 1}}
  \mathbf{r}^l\|_2 \geq \|\mathcal{P}_{\Lambda^{l + 1}} \mathbf{r}^l\|_2.
\end{eqnarray}
Since $\mathcal{P}_{\Lambda^{l + 1}} = \mathcal{P}'_{\Lambda^{l +
1}} = ( \mathbf{\Phi}_{\Lambda^{l + 1}}^\dag)'
\mathbf{\Phi}'_{\Lambda^{l + 1}}$, we further have
\begin{eqnarray}
   \|\mathbf{r}^l - \mathbf{r}^{l + 1}\|_2 &\geq& \|(\mathbf{\Phi}_{\Lambda^{l + 1}}^\dag)' \mathbf{\Phi}'_{\Lambda^{l + 1}} \mathbf{r}^l\|_2 \nonumber \\
&\geq& (1 + \delta_{S})^{-1/2} \|\mathbf{\Phi}'_{\Lambda^{l + 1}}
\mathbf{r}^l\|_2 \label{eq:mmmss}
\end{eqnarray}
where \eqref{eq:mmmss} is because the singular values
of $\mathbf{\Phi}_{\Lambda^{l + 1}}$ lie between $(1 -
\delta_{S})^{1/2}$ and $(1 + \delta_{S})^{1/2}$ and hence the
smallest singular value of $\mathbf{\Phi}_{\Lambda^{l + 1}}^\dag$ is
lower bounded by $(1 + \delta_{S})^{-1/2}$.\footnote{Suppose the
matrix $\mathbf{\Phi}_{\Lambda^{l + 1}}$ has singular value
decomposition $\mathbf{\Phi}_{\Lambda^{l + 1}} = \mathbf{U}
\mathbf{\Sigma} \mathbf{V}'$, then $\mathbf{\Phi}_{\Lambda^{l +
1}}^\dag = \mathbf{V} \mathbf{\Sigma}^\dag \mathbf{U}'$ where
$\mathbf{\Sigma}^\dag$ is the pseudoinverse of $\mathbf{\Sigma}$,
which is formed by replacing every non-zero diagonal entry by its
reciprocal and transposing the resulting matrix.}

\vspace{2mm}

\item {\bf Proof of \eqref{eq:nalia}}:

We first introduce a lemma useful in our proof.

\vspace{1mm}
\begin{lemma} \label{lem:good}
 Let $\mathbf{u}, \mathbf{z} \in \mathcal{R}^{n}$ be two distinct
vectors and let $W = supp(\mathbf{u}) \cap supp(\mathbf{z})$. Also,
let $U$ be the set of $S$ indices corresponding to $S$ most
significant elements
  in $\mathbf{u}$. Then for any integer $S \geq 1$,
  \begin{equation}
    \langle \mathbf{u}, \mathbf{z} \rangle \leq \left( \left\lceil \frac{|W|}{S} \right\rceil\right)^{1/2} \|\mathbf{u}_{U}\|_2 \|\mathbf{z}_W\|_2.
  \end{equation}
\end{lemma}
\begin{proof}
See Appendix \ref{app:lemmaD}.
\end{proof}

\vspace{1mm}
Now we are ready to prove \eqref{eq:nalia}. Let $\mathbf{u} =
\mathbf{\Phi}' \mathbf{r}^l$ and let $\mathbf{z} \in \mathcal{R}^n$
be the vector satisfying $\mathbf{z}_{T \cap T^{k} \cup
{\Gamma}^k_\tau} = \mathbf{x}_{T \cap T^{k} \cup {\Gamma}^k_\tau}$
and $\mathbf{z}_{\Omega \backslash (T \cap T^{k} \cup
{\Gamma}^k_\tau)} = \mathbf{0}$.  Since $supp(\mathbf{u}) = \Omega
\backslash T^l$ and $supp(\mathbf{z}) = T \cap T^{k} \cup
{\Gamma}^k_\tau$ and also noting that $T^k \subseteq T^l$, we have
$W = supp(\mathbf{u}) \cap supp(\mathbf{z}) = {\Gamma}^k_\tau
\backslash T^l$. Moreover, since $\Lambda^{l + 1}$ contains the
indices corresponding to $S$ most significant elements in
$\mathbf{u} = \mathbf{\Phi}' \mathbf{r}^l$, we have $U = \Lambda^{l
+ 1}$.   Using Lemma \ref{lem:good},
  \begin{eqnarray}
 \langle \mathbf{\Phi}' \mathbf{r}^l, \mathbf{z} \rangle
&\leq& \left(\left\lceil \frac{|{\Gamma}^k_\tau \backslash T^l|}{S} \right\rceil\right)^{1/2}
\|\mathbf{\Phi}_{\Lambda^{l + 1}}' \mathbf{r}^l \|_2
\|\mathbf{z}_{{\Gamma}^k_\tau \backslash T^l}\|_2 \nonumber \\
&\leq& \left(\left\lceil \frac{|{\Gamma}^k_\tau |}{S} \right\rceil\right)^{1/2}
\|\mathbf{\Phi}_{\Lambda^{l + 1}}' \mathbf{r}^l \|_2
\|\mathbf{z}_{{\Gamma}^k_\tau \backslash T^l}\|_2 \nonumber \\
&\leq& \left(\left\lceil \frac{|{\Gamma}^k_\tau |}{S} \right\rceil\right)^{1/2}
\|\mathbf{\Phi}_{\Lambda^{l + 1}}' \mathbf{r}^l \|_2
\|\mathbf{z}_{\Omega \backslash T^l}\|_2. \label{eq:youmyou}
  \end{eqnarray}

On the other hand,
\begin{eqnarray}
~&&\hspace{-8.5mm} \langle \mathbf{\Phi}' \mathbf{r}^l, \mathbf{z} \rangle \nonumber \\
&&\hspace{-7mm} = 
\langle \mathbf{\Phi}' \mathbf{r}^l, \mathbf{z} - \hat{\mathbf{x}}^l
\rangle + \langle
\mathbf{\Phi}' \mathbf{r}^l, \hat{\mathbf{x}}^l \rangle \nonumber \\
&&\hspace{-7mm} \overset{(a)}{=} \langle \mathbf{\Phi}' \mathbf{r}^l,
\mathbf{z} - \hat{\mathbf{x}}^l \rangle \nonumber \\
&&\hspace{-7mm} = \langle \mathbf{\Phi}(\mathbf{z} - \hat{\mathbf{x}}^l), \mathbf{r}^l \rangle  \nonumber \\
&&\hspace{-7mm} \overset{(b)}{=} \frac{1}{2}\left(\|\mathbf{\Phi}(\mathbf{z} -
\hat{\mathbf{x}}^l)\|_2^2  + \|\mathbf{r}^l\|_2^2 - \|\mathbf{r}^l -
\mathbf{\Phi}(\mathbf{z} - \hat{\mathbf{x}}^l)\|_2^2
  \right) \nonumber \\
&&\hspace{-7mm} \overset{(c)}{=} \frac{1}{2}\left(\|\mathbf{\Phi}(\mathbf{z} - \hat{\mathbf{x}}^l)\|_2^2  + \|\mathbf{r}^l\|_2^2 - \|\mathbf{\Phi}(\mathbf{x} - \mathbf{z}) + \mathbf{v}\|_2^2   \right)  \nonumber \\
&&\hspace{-7mm} = \frac{1}{2}\left(\|\mathbf{\Phi}(\mathbf{z} -
\hat{\mathbf{x}}^l)\|_2^2 + \|\mathbf{r}^l\|_2^2 -
\|\mathbf{\Phi}_{{\Gamma^{k}}\backslash {\Gamma}^k_\tau}
\mathbf{x}_{{\Gamma^{k}} \backslash {\Gamma}^k_\tau} +
\mathbf{v}\|_2^2 \right) 
\nonumber \\
&&\hspace{-7mm} \overset{(d)}{\geq} \|\mathbf{\Phi}(\mathbf{z} -
\hat{\mathbf{x}}^l) \|_2 \left({\|\mathbf{r}^l \|_2^2 -
\|\mathbf{\Phi}_{{\Gamma^{k}} \backslash {\Gamma}^k_\tau}
\mathbf{x}_{{\Gamma^{k}} \backslash {\Gamma}^k_\tau} +
\mathbf{v}\|_2^2}\right)^{1/2}
\nonumber \\
&&\hspace{-7mm} \overset{(e)}{\geq} (1 -
  \delta_{|{\Gamma}^k_{\tau} \cup T^l|})^{1/2} ~ \|\mathbf{z} -
  \hat{\mathbf{x}}^l\|_2  \nonumber \\
  &&\hspace{-2mm} \times \left({\|\mathbf{r}^l \|_2^2 -
\|\mathbf{\Phi}_{{\Gamma^{k}} \backslash {\Gamma}^k_\tau}
\mathbf{x}_{{\Gamma^{k}} \backslash {\Gamma}^k_\tau} +
\mathbf{v}\|_2^2}\right)^{1/2}
\nonumber \\
&&\hspace{-7mm} \geq (1 - \delta_{|{\Gamma}^k_{\tau} \cup T^l|})^{1/2} ~
\|(\mathbf{z} - \hat{\mathbf{x}}^l)_{\Omega \backslash
  T^l}\|_2  \nonumber \\
  &&\hspace{-2mm} \times \left({\|\mathbf{r}^l \|_2^2 -
\|\mathbf{\Phi}_{{\Gamma^{k}} \backslash {\Gamma}^k_\tau}
\mathbf{x}_{{\Gamma^{k}} \backslash {\Gamma}^k_\tau} +
\mathbf{v}\|_2^2}\right)^{1/2} \nonumber \\
&&\hspace{-7mm} \overset{(f)}{\geq} (1 - \delta_{|{\Gamma}^k_{\tau} \cup
T^l|})^{1/2}  ~   \|\mathbf{z}_{\Omega \backslash T^l}\|_2  \nonumber \\
  &&\hspace{-2mm} \times \left({\|\mathbf{r}^l \|_2^2 -
\|\mathbf{\Phi}_{{\Gamma^{k}} \backslash {\Gamma}^k_\tau}
\mathbf{x}_{{\Gamma^{k}} \backslash {\Gamma}^k_\tau} +
\mathbf{v}\|_2^2}\right)^{1/2},   \label{eq:zuobianyou}
\end{eqnarray}
where (a) is because $supp(\hat{\mathbf{x}}^l) = T^l$ and
$supp(\mathbf{\Phi}' \mathbf{r}^l) = \Omega \backslash T^l$ and
hence $\langle \mathbf{\Phi}' \mathbf{r}^l, \hat{\mathbf{x}}^l
\rangle = \mathbf{0}$, (b) uses the fact that $\langle \mathbf{u}, \mathbf{v}\rangle = \frac{1}{2} (\|\mathbf{u}\|_2^2 + \|\mathbf{v}\|_2^2 - \|\mathbf{u} - \mathbf{v}\|_2^2)$, (c) is from $\mathbf{r}^l + \mathbf{\Phi}
\hat{\mathbf{x}}^l = \mathbf{y} = \mathbf{\Phi} \mathbf{x} +
\mathbf{v}$, (d) uses the inequality $\frac{1}{2}(a + b) \geq
\sqrt{ab}$ (with $a = \|\mathbf{\Phi}(\mathbf{z} -
\hat{\mathbf{x}}^l)\|_2^2$ and $b = \|\mathbf{r}^l\|_2^2 -
\|\mathbf{\Phi}_{{\Gamma^{k}}\backslash {\Gamma}^k_\tau}
\mathbf{x}_{{\Gamma^{k}} \backslash {\Gamma}^k_\tau} +
\mathbf{v}\|_2^2$),\footnote{Note that we only need to consider the
case $\|\mathbf{r}^l\|_2^2 - \|\mathbf{\Phi}_{{\Gamma^{k}}
\backslash {\Gamma}^k_\tau} \mathbf{x}_{{\Gamma^{k}} \backslash
{\Gamma}^k_\tau} + \mathbf{v}\|_2^2 \geq 0$. For the alternative
case $\|\mathbf{r}^l\|_2^2 - \|\mathbf{\Phi}_{{\Gamma^{k}}
\backslash {\Gamma}^k_\tau} \mathbf{x}_{{\Gamma^{k}} \backslash
{\Gamma}^k_\tau} + \mathbf{v}\|_2^2 < 0$, \eqref{eq:nalia} directly
holds true since $\|\mathbf{\Phi}'_{\Lambda^{l + 1}}
\mathbf{r}^l\|_2^2 \geq 0$.} (e) is from the RIP ($\|\mathbf{z} -
\hat{\mathbf{x}}^l\|_0 = |{\Gamma}^k_{\tau} \cup T^l|$), and (e) is
due to $(\hat{\mathbf{x}}^l)_{\Omega \backslash T^l} = \mathbf{0}$.
\vspace{2mm}
Finally, using \eqref{eq:youmyou} and \eqref{eq:zuobianyou}, we have
\begin{eqnarray}
\lefteqn{\|\mathbf{\Phi}'_{\Lambda^{l + 1}} \mathbf{r}^l\|_2}
\nonumber \\
&& \hspace{-6mm} \geq \left(\hspace{-.5mm} {\frac{1 - \delta_{|{\Gamma}^k_{\tau}
\cup T^l |}} {\left\lceil \frac{|{\Gamma}^k_\tau |}{S} \right\rceil}}
\left(\|\mathbf{r}^l \|_2^2 - \|\mathbf{\Phi}_{{\Gamma^{k}}
\backslash {\Gamma}^k_\tau} \mathbf{x}_{{\Gamma^{k}} \backslash
{\Gamma}^k_\tau} + \mathbf{v}\|_2^2 \right)\hspace{-.5mm} \right)^{\hspace{-1mm}1/2}, \nonumber
\end{eqnarray}
which is the desired result.
\end{itemize}

 \end{proof}

\section{Proof of Lemma \ref{lem:good}} \label{app:lemmaD}

\begin{proof}
We consider two cases: 1) $1 \leq S \leq |W|$ and 2) $S > |W|$. 

We first consider the case $1 \leq S \leq |W|$. Without loss of
generality, we assume that $W = \left\{1, 2, \cdots, |W|\right\}$ and that the
elements of $\mathbf{u}_W$ are arranged in descending order of
their magnitudes. We define the subset $W_i$ of $W$ as
\begin{equation}
 W_i \hspace{-.5mm}= \hspace{-.5mm}
\begin{cases}
\left\{S(i - 1) + 1, \cdots, Si \right\} &\hspace{-1.5mm} i = 1, \cdots\hspace{-.5mm},\hspace{-.25mm} \left\lceil \frac{|W|}{S} \right\rceil\hspace{-.5mm} - \hspace{-.5mm}1, \\
\left\{S \left(\left\lceil \frac{|W|}{S} \right\rceil - 1\right) + 1, \cdots\hspace{-.5mm},\hspace{-.25mm} |W|\right\} &\hspace{-1.5mm}i = \lceil
\frac{|W|}{S} \rceil.
\end{cases}  \label{eq:jjjjffff0}
\end{equation} 
%
See Fig.~\ref{fig:set_3} for the illustration of indices in $W_i$.
Note that when $\left\lceil \frac{|W|}{S} \right\rceil > \frac{|W|}{S}$, the last
set $W_{\left\lceil \frac{|W|}{S} \right\rceil}$ has less than $S$ elements.
\begin{figure}[t]
\begin{center}
\includegraphics[width = 75 mm]{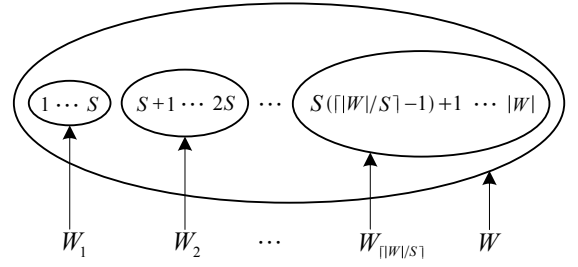}
\caption{Illustration of indices in $W_i$.} \label{fig:set_3}
\end{center}
\end{figure}

Observe that
  \begin{equation}
    \langle \mathbf{u}, \mathbf{z} \rangle = \langle \mathbf{u}_W, \mathbf{z}_W
    \rangle
    \leq \sum_i |\langle \mathbf{u}_{W_i}, \mathbf{z}_{W_i} \rangle|
    \leq \sum_i \|\mathbf{u}_{W_i}\|_2 \|\mathbf{z}_{W_i}\|_2,
  \end{equation}
where the second inequality is due to the H\"{o}lder's inequality.
By the definition of $U$, we have $\|\mathbf{u}_{U}\|_2 \geq
\|\mathbf{u}_{W_1}\|_2 = \max_i \|\mathbf{u}_{W_i}\|_2$ and hence
  \begin{eqnarray}
    \langle \mathbf{u}, \mathbf{z} \rangle
    &\leq& \|\mathbf{u}_{U}\|_2 \sum_i  \|\mathbf{z}_{W_i}\|_2 \label{eq:2233441} \\
    &\leq& \|\mathbf{u}_{U}\|_2 \left(\left\lceil \frac{|W|}{S} \right\rceil \sum_i  \|\mathbf{z}_{W_i}\|_2^2 \right)^{1/2} \label{eq:2233442} \\
    &=& \left(\left\lceil \frac{|W|}{S} \right\rceil \right)^{1/2} \|\mathbf{u}_{U}\|_2
    \|\mathbf{z}_W\|_2,
  \end{eqnarray}
where \eqref{eq:2233442} follows from the fact that $\sum_{i =
1}^{d} a_i \leq \left(d \sum_{i = 1}^{d} a_i^2 \right)^{1/2}$ with
$a_i = \|\mathbf{z}_{W_i}\|_2$ and $d = \left\lceil \frac{|W|}{S} \right\rceil$.

Now, we consider the alternative case ($S > |W|$). In this
case, it is clear that $\left(\left\lceil \frac{|W|}{S} \right\rceil \right)^{1/2} = 1$ and
$\|\mathbf{u}_{U}\|_2 \geq \|\mathbf{u}_{W}\|_2$, and hence
\begin{eqnarray}
  \left(\left\lceil \frac{|W|}{S} \right\rceil \right)^{1/2} \|\mathbf{u}_{U}\|_2
  \|\mathbf{z}_W\|_2 &=& \|\mathbf{u}_{U}\|_2
  \|\mathbf{z}_W\|_2 \nonumber \\
  &\geq& \|\mathbf{u}_{W}\|_2
  \|\mathbf{z}_W\|_2 \nonumber \\
  &\geq& \langle \mathbf{u}_W, \mathbf{z}_W \rangle \label{eq:h} \nonumber \\
  &=& \langle \mathbf{u}, \mathbf{z} \rangle,
\end{eqnarray}
which completes the proof.   
\end{proof}

\section{Proof of Proposition \ref{prop:residual}} \label{app:prop2}
\begin{proof} 
Recall from Proposition \ref{prop:rk} that for
given $\Gamma^{k}$ and any integer $l \geq k$, the residual of gOMP
satisfies
\begin{eqnarray}  
\lefteqn{ \|\mathbf{r}^l\|_2^2 - \|\mathbf{r}^{l + 1}\|_2^2 \geq
\frac{1 - \delta_{|{\Gamma}^k_{\tau} \cup T^l |}} {(1 + \delta_{S})
\left\lceil \frac{|{\Gamma}^k_\tau |}{S} \right\rceil}} \nonumber \\
&& ~~~~~\times \left( \|\mathbf{r}^l \|_2^2 -
\|\mathbf{\Phi}_{{\Gamma^{k}} \backslash {\Gamma}^k_\tau}
\mathbf{x}_{{\Gamma^{k}} \backslash {\Gamma}^k_\tau} + \mathbf{v}
\|_2^2 \right),~~~~~\label{eq:residual3qq}
\end{eqnarray}
where $\tau = 1, 2, \cdots, \max \left\{0, \left\lceil \log_2 \frac{|\Gamma^{k}|}{S} \right\rceil \right\} + 1$. 
Since $a > 1 - \exp(- a)$ for $a > 0$ and
$$\frac{1 - \delta_{|{\Gamma}^k_{\tau} \cup T^l |}} {\left\lceil
\frac{|{\Gamma}^k_\tau |}{S} \right\rceil(1 + \delta_{S})} > 0,$$ we have
\begin{equation}
 \frac{1 - \delta_{|{\Gamma}^k_{\tau} \cup T^l |}} {\left\lceil
\frac{|{\Gamma}^k_\tau |}{S} \right\rceil(1 + \delta_{S})} \geq 1 - \exp \left( -
\frac{1 - \delta_{|{\Gamma}^k_{\tau} \cup T^l |}} {\left\lceil
\frac{|{\Gamma}^k_\tau |}{S} \right\rceil(1 + \delta_{S})}
 \right). \label{eq:youma?0}
\end{equation}
Using \eqref{eq:residual3qq} and \eqref{eq:youma?0},
\begin{eqnarray} \label{eq:residual3good}
 \lefteqn{ \|\mathbf{r}^l\|_2^2 - \|\mathbf{r}^{l + 1}\|_2^2 \geq  \left(1
- \exp \left( - \frac{1 - \delta_{|{\Gamma}^k_{\tau} \cup T^l |}}
{\left\lceil
\frac{|{\Gamma}^k_\tau |}{S} \right\rceil(1 + \delta_{S})}
 \right) \right)} \nonumber \\
&& ~~~~~~~~~~~~~~~~\times  \left( \|\mathbf{r}^l \|_2^2 -
\|\mathbf{\Phi}_{{\Gamma^{k}} \backslash {\Gamma}^k_\tau}
\mathbf{x}_{{\Gamma^{k}} \backslash {\Gamma}^k_\tau} + \mathbf{v}
\|_2^2 \right).~~~~~~
\end{eqnarray}
Subtracting both sides of \eqref{eq:residual3good} by $\|
\mathbf{r}^l \|_2^2 - \| \mathbf{\Phi}_{\Gamma^k \backslash
{\Gamma}^k_\tau} \mathbf{x}_{\Gamma^k \backslash {\Gamma}^k_\tau} +
 \mathbf{v}\|_2^2$, we have
\begin{eqnarray}
 \lefteqn{\|\mathbf{r}^{l + 1}  \|_2^2 - \| \mathbf{\Phi}_{{\Gamma^{k}}
\backslash {\Gamma}^k_\tau} \mathbf{x}_{{\Gamma^{k}} \backslash
{\Gamma}^k_\tau} +
 \mathbf{v}\|_2^2} \nonumber \\
&& ~~~~~~~~~~~~\leq \exp \left( - \frac{1 -
\delta_{|{\Gamma}^k_{\tau} \cup T^l
|}} {\left\lceil
\frac{|{\Gamma}^k_\tau |}{S} \right\rceil(1 + \delta_{S})} \right)  \nonumber \\
&& ~~~~~~~~~~~~~~~\times (\| \mathbf{r}^l  \|_2^2 - \|
\mathbf{\Phi}_{{\Gamma^{k}} \backslash {\Gamma}^k_\tau}
\mathbf{x}_{{\Gamma^{k}}
 \backslash {\Gamma}^k_\tau} + \mathbf{v} \|_2^2), \nonumber
\end{eqnarray}
and also
\begin{eqnarray}
 \lefteqn{\| \mathbf{r}^{l + 2}  \|_2^2 -  \| \mathbf{\Phi}_{{\Gamma^{k}}
\backslash {\Gamma}^k_\tau} \mathbf{x}_{{\Gamma^{k}} \backslash
{\Gamma}^k_\tau}  + \mathbf{v}\|_2^2} \nonumber \\
&& ~~~~~~~~~~~~\leq \exp \left( - \frac{1 -
\delta_{|{\Gamma}^k_{\tau} \cup T^{l + 1} |}} {\left\lceil
\frac{|{\Gamma}^k_\tau |}{S} \right\rceil(1 + \delta_{S})} \right) \nonumber \\
&& ~~~~~~~~~~~~~~~\times (\|\mathbf{r}^{l + 1}  \|_2^2 -  \|
\mathbf{\Phi}_{{\Gamma^{k}} \backslash {\Gamma}^k_\tau}
\mathbf{x}_{{\Gamma^{k}} \backslash {\Gamma}^k_\tau}   +
 \mathbf{v}\|_2^2),\nonumber \\
 &&~~~~~~~~~~~~~\vdots  \nonumber \\
 \lefteqn{\| \mathbf{r}^{l + \Delta l}  \|_2^2 -  \|
\mathbf{\Phi}_{{\Gamma^{k}} \backslash {\Gamma}^k_\tau}
\mathbf{x}_{{\Gamma^{k}} \backslash {\Gamma}^k_\tau}  +
\mathbf{v}\|_2^2} \nonumber \\
&& ~~~~~~~~~~~~\leq \exp \left( - \frac{1 -
\delta_{|{\Gamma}^k_{\tau} \cup T^{l + \Delta l - 1} |}} {\left\lceil
\frac{|{\Gamma}^k_\tau |}{S} \right\rceil(1 + \delta_{S})} \right) \nonumber \\
&& ~~~~~~~~~~~~~~~\times (\| \mathbf{r}^{l + \Delta l - 1}  \|_2^2 -
{ \| \mathbf{\Phi}_{{\Gamma^{k}} \backslash {\Gamma}^k_\tau}
\mathbf{x}_{{\Gamma^{k}} \backslash
   {\Gamma}^k_\tau}  + \mathbf{v}
   \|_2^2}). \nonumber
\end{eqnarray}
Some additional manipulations yield the following result.
\begin{eqnarray}
 \lefteqn{\| \mathbf{r}^{l + \Delta l}  \|_2^2 -  \|
\mathbf{\Phi}_{{\Gamma^{k}} \backslash {\Gamma}^k_\tau}
\mathbf{x}_{{\Gamma^{k}} \backslash {\Gamma}^k_\tau}  +
\mathbf{v}\|_2^2} \nonumber \\
&& ~~~~~~~~~~~~\leq \prod_{i = l}^{l + \Delta l - 1} \exp \left( -
\frac{1 - \delta_{|{\Gamma}^k_{\tau} \cup T^i |}} {\left\lceil
\frac{|{\Gamma}^k_\tau |}{S} \right\rceil(1 + \delta_{S})} \right) \nonumber \\
&& ~~~~~~~~~~~~~~~\times
 (\| \mathbf{r}^{l}  \|_2^2 - { \|
\mathbf{\Phi}_{{\Gamma^{k}} \backslash {\Gamma}^k_\tau}
\mathbf{x}_{{\Gamma^{k}} \backslash {\Gamma}^k_\tau}
    + \mathbf{v} \|_2^2}). \nonumber
\end{eqnarray}
Since $\delta_{ | {\Gamma}^k_\tau \cup T^{l} |} \leq \delta_{ |
{\Gamma}^k_\tau \cup T^{l + 1} |} \leq \cdots \leq \delta_{ |
{\Gamma}^k_\tau \cup T^{l + \Delta l - 1} |}$, we further have
\begin{eqnarray}
\lefteqn{\| \mathbf{r}^{l + \Delta l} \|_2^2 -
\|\mathbf{\Phi}_{{\Gamma^{k}} \backslash {\Gamma}^k_\tau}
\mathbf{x}_{{\Gamma^{k}} \backslash {\Gamma}^k_\tau} +
\mathbf{v}\|_2^2} \nonumber \\
&& ~~~~~~~~~\leq C_{\tau,l, \Delta l} \left(\| \mathbf{r}^l \|_2^2 +
\|\mathbf{\Phi}_{{\Gamma^{k}} \backslash {\Gamma}^k_\tau}
\mathbf{x}_{{\Gamma^{k}} \backslash {\Gamma}^k_\tau} +
\mathbf{v}\|_2^2\right), \nonumber
  \label{eq:good1}
\end{eqnarray}
where
\begin{equation}
C_{\tau,l, \Delta l} = \exp \left( - \frac{\Delta l (1 - \delta_{|
{\Gamma}^k_\tau \cup T^{l + \Delta l - 1}|})} {\left\lceil
\frac{|{\Gamma}^k_\tau |}{S} \right\rceil (1 + \delta_{S})} \right),
\label{eq:constantc}
\end{equation}
which completes the proof.   
\end{proof}

\bibliographystyle{IEEEbib}
\bibliography{CS_refs}

\end{document}